\documentclass[a4paper,10pt]{article}
\usepackage{amssymb,amsmath,amsthm}
\usepackage{amsfonts}
\usepackage[english]{babel}
\usepackage{hyperref}
\usepackage{verbatim} 
\usepackage{bbm} 
\usepackage{color}
\numberwithin{equation}{section}
\usepackage{enumitem}
\usepackage{relsize}
\usepackage{xcolor}

\newcommand{\Om}{\Omega}

\newcommand{\Op}{\mathrm{Op}\,}

\newcommand{\Lip}{\mathrm{Lip}}

\newtheorem{theorem}{Theorem}[section]
\newtheorem{proposition}[theorem]{Proposition}
\newtheorem{lemma}[theorem]{Lemma}
\newtheorem{corollary}[theorem]{Corollary}

\newtheorem{remark}[theorem]{Remark}
\newtheorem{remarks}[theorem]{Remark}
\newtheorem{definition}[theorem]{Definition}
\newcommand{\be}{\begin{equation}}
\newcommand{\ee}{\end{equation}}

\newcommand{\om}{\omega}

\newcommand{\e}{\varepsilon}

\newcommand{\R}{\mathbb R}
\newcommand{\C}{\mathbb C}

\newcommand{\Z}{\mathbb Z}
\newcommand{\Zd}{\mathbb Z^d}
\newcommand{\Zp}{\mathbb Z^p}

\newcommand{\N}{\mathbb N}
\newcommand{\T}{\mathbb T}

\newcommand{\s }{\sigma }
\newcommand{\ii }{{\rm i} }

\newcommand{\g }{\gamma}

\newcommand{\vphi}{\varphi }

\newcommand{\cW}{\mathcal{W}}

\newcommand{\cU}{\mathcal{U}}
\newcommand{\cO}{\mathcal{O}}

\def\togli#1{\relax}

\def\muzero{\nu^{(0)}}

\def\ba{\begin{aligned}}
\def\ea{\end{aligned}}
\def\beginm{\begin{multline}}
\def\endm{\end{multline}}

\newcommand{\Lipg}{{\rm{Lip}(\g)}}

\newcommand{\csi}{\xi}
\newcommand{\frake}{{\frak e}}

\newcommand{\calg}{{\cal G}}
\newcommand{\calw}{{\cal W}}

\newcommand{\omphi}{\omega \cdot \partial_{\vphi}}
\newcommand{\tildeom}{\tilde{\omega}}
\newcommand{\partialt}{\partial_{t}}

\newcommand{\Gradbeta}{\langle \nabla \rangle^{\beta}}
\newcommand{\cinfty}{{\cal C}^{\infty}}

\newcommand{\Omegauno}{{\Omega}_{0, \gamma}}
\DeclareMathOperator{\diag}{diag}

\newcommand{\four}[3]{\widehat{#1}_{#2 #3}}

\newcommand{\pediceSunoSdue}[2]{ _{{\cal B}(H^{ #1}, H^{ #2})} }

\newcommand{\indSsigsig}[3]{ _{{\cal M}^{#1}_{ #2, #3}}}
\newcommand{\normaSsigsig}[3]{ \| {#1} \|^{\Lip}_{{\cal M}^{s}_{ #2, #3}}}
\newcommand{\normabeta}[3]{ \| {#1} \|^{\Lip}_{{\cal W}^{s, \beta}_{ #2, #3}}}
\newcommand{\normahs}[3]{\| {#1} \|^{HS}_{ #2, #3}}

\newcommand{\bSunoSdue}[2]{{\cal B}\left({\cal H}^{#1} ,  {\cal H}^{#2}\right)}

\newcommand{\bhsSunoSdue}[2]{{\cal B}^{HS}({\cal H}^{#1}, {\cal H}^{#2})}
\newcommand{\Lipsp}{{\mathcal{L}ip}}

\begin{document}

\title{{\bf Reducibility of non-resonant transport equation on $\T^d$ with unbounded perturbations}}

\date{}

\author{Dario Bambusi, Beatrice Langella, Riccardo Montalto \footnote{Supported in part by the Swiss National Science Foundation}}

\maketitle

\begin{abstract}
We prove reducibility of a transport equation on the $d$-dimensional torus $\T^d$ with a time
quasi-periodic unbounded perturbation. As far as we know this is the
first example of a reducibility result for an equation in more than one
dimensions with unbounded perturbations. Furthermore the unperturbed
problem has eigenvalues whose differences are dense on the real axis. 
\end{abstract}

\noindent




\section{Introduction}
In this paper we obtain reducibility for a transport equation on the $d$-dimensional torus $\T^d$, $\T := \R /(2 \pi \Z)$, $d \geq 1$ of the form 
\begin{equation}\label{main equation}
\partial_t u = \Big( \nu + \e V(\omega t, x) \Big) \cdot \nabla u + \e {\cal W}(\omega t)[u], 
\end{equation}
where the frequencies $\omega \in \R^{\frak n}$, and $\nu \in \R^d$
play the role of parameters, 
$\e > 0$ is a small parameter, $V \in {\cal C}^\infty(\T^{\frak n}
\times \T^d, \R^d)$ is a real function and ${\cal W}(\vphi)$, $\vphi \in \T^{\frak n}$ is
a pseudo-differential operator of order $1 - \frak e$, for some $\frak
e > 0$. More precisely our aim is to show that for $\e$ small enough
and for {\it most} values of $\widetilde \omega =
(\omega, \nu) \in \Omega := [1, 2]^{\frak n + d}$, there exists a
bounded and invertible transformation (acting on the scale of Sobolev
spaces) which transforms the PDE \eqref{main equation} into another one
whose vector field is a time independent diagonal operator.

This is the first example of a reducibility result for unbounded
perturbations of a Hamiltonian PDE in more than one space
dimension. Furthermore, the unperturbed problem has eigenvalues whose
differences are dense on the real axis, a case which is usually
considered as particular difficult to deal with. 

Following \cite{BBM14} (see also \cite{BertiMontalto, Bam18, Bam17, Mon17a,BBHM}), the proof consists of two steps:
first we use pseudo-differential calculus in order to transform the
original system to a system with a smoothing perturbation (smoothing
theorem) and then we apply a KAM scheme in order to actually obtain
reducibility. The smoothing theorem is obtained through a variant of the theory
developed in \cite{BGMR2} and the KAM theory is a variant of the one
developed in \cite{BBHM}. The main purpose of the present paper is to show
that it is possible to glue together such tools in order to deal
with a nontrivial higher dimensional problem. The main technical
difficulty consists in
showing that the frequencies $(\omega,\nu)$ can be used to tune the
small divisors and to fulfill some second Melnikov type nonresonance
conditions.

A further novelty is that, in the equation \eqref{main equation}, it is natural to
consider perturbations $\cW$ s.t. i$\cW$ is not a symmetric operator,
so we consider the case where i$\cW$ is only symmetric hyperbolic
(namely that ${\cal W} + {\cal W}^* $ is an operator of order $0$, see
Definition \ref{stru} below) and, in order to get information on the
behavior of the solutions, we study also the case where it has some
additional structures, namely reality and reversibility (see
Definition \ref{stru} below). In this case
we also get the stability, namely all the Sobolev norms of the
solutions of the equation \eqref{main equation} stay bounded for all
times. Note that by Corollary \ref{corollario caso rev real}, in the
non-reversible case, one can construct solutions whose Sobolev
norms go to infinity.
\vskip10pt

There is a wide literature on the dynamics of time periodic or
quasiperiodic Schr\"odinger type equations, starting from the
pioneering works \cite{Bel,C87} (see also \cite{DS96}). Concerning the
problem of reducibility, we just mention \cite{Kuk93,BG01,LY10}, in
which the classical methods developed in KAM theory (in particular
\cite{K87,Kuk97}) have been adapted
and extended in order to deal with the case where the unperturbed
equation has order $n$ and the perturbation is of order $\delta \leq n
- 1$. All these results are for equations in one space dimension. 
\togli{the first 
results have been obtained as a corollary of KAM
theorems (see \cite{K87,W90,Kuk93}). Also concerning unbounded
perturbations, the first reducibility results (see \cite{BG01}) have
been obtained by applying the theory that Kuksin developed in order
to deal with KdV equation (see \cite{Kuk97},
see also \cite{KapPoe}). These results apply to the
case where the unperturbed equation has order $n$ and the perturbation
is of order $\delta < n - 1$. The extension to the case $\delta=n-1$
was then obtained in \cite{LY10}.}

The breakthrough for further developments was obtained in
\cite{BBM14}, developing ideas introduced in \cite{IPT05}. The
strategy introduced in \cite{BBM14} is based on the usage of
pseudo-differential calculus, which allows to reduce the order of the
perturbation, before applying reducibility schemes based on KAM
theory. In particular their method allows to reduce the original
problem to a problem in which the perturbation is a smoothing operator
of arbitrary order.  These ideas have been applied in the field of KAM
theory for one dimensional PDEs by several authors (see
\cite{BBM-auto,BBM-mKdV,FP15,BertiMontalto,Mon17a,Bam18,Bam17}) and
the extension to some particular models in more than one dimension has
also been obtained \cite{BGMR1,Mon17a}.

The idea of using pseudo-differential calculus in order to conjugate
the original system to another one with a smoothing perturbation has
shown to be very useful, also in control theory, see
\cite{controlloAB,controlloBHF,controlloBHM} and in the problem of
estimating the growth of the Sobolev norms
\cite{BGMR2,Mon18,Montalto2}.

Actually, the methods developed in \cite{BGMR2} are the starting point
of the present paper.

The second kind of ideas on which we rely were developed in
\cite{BBHM} (and extended in \cite{Mon17a}) where the
authors developed a reducibility scheme for smoothing perturbation of
a system whose
frequencies fulfill very bad nonresonance conditions (see
eq. \eqref{bad} below). The idea is that the smoothing character of
the nonlinearity can be used to recover a smoothness loss due to the
small denominators. In \cite{BBHM}, the method was applied to the case
where the frequencies of the linear system grow at infinity like
$\omega_j = j^{1/2}$, $j \in \N$. Here we adapt the scheme to the case
where the differences between couples of frequencies are dense on the
real axis.

We recall that previous reducibility results in higher dimensional
systems have been obtained only in cases where the frequencies of the
unperturbed system have a very particular structure \cite{EK09,GP} so
that the more or less standard second order Melnikov conditions can be
imposed {\it blockwise}.

\noindent
The paper is organized as follows. In Section \ref{sta} we state
precisely our main theorem. In Section \ref{sezione regolarizzazione
  grande} we conjugate the vector field of the equation \eqref{main
  equation} to another one which is an arbitrarily smoothing
perturbation of a diagonal operator. The reduction to constant
coefficients of the highest order is implemented in Section
\ref{sezione riduzione ordine alto} (following \cite{trasporto paper}).

In Section \ref{sezione regolarizzazione} we reduce to constant
coefficients the lower order terms up to an arbitrarily smoothing
remainder (following \cite{BGMR2}). In
the present paper, such a procedure is implemented by assuming only
that the remainders arising at each step are symmetric hyperbolic. 

\noindent
In Section \ref{sezione riducibilita grande} we perform a
KAM-reducibility scheme for vector fields which are smoothing
perturbations of a diagonal one, by imposing second order Melnikov
conditions with loss of derivatives in space (see Theorem
\ref{thm:abstract linear reducibility}). Note that the {\it final
  eigenvalues} $\lambda_j^{(\infty)}$, appearing in the definition of
the set \eqref{Cantor finale} (on which you get the diagonalization)
have an asymptotic expansion of the form
\begin{equation}
  \label{lafina}
\lambda_j^{(\infty)}= \ii \nu^{(0)} \cdot j + z(j) + O (\e \langle j \rangle^{- 2m})
\end{equation}
for some $m > 0$ large enough, where $\nu^{(0)}$ is a constant vector,
$z$ is a Fourier multiplier of order $1 - \frak e$. The
  fact that $z$ is a pseudo-differential operator is used in the
  measure estimate of Section \ref{sezione stime misura}, in
  particular, in Lemma \ref{stima singolo risonante} to obtain the
  estimate $|z(j) - z(j')| \lesssim \e |j - j'|$ for any $j, j' \in
  \Z^d$. In \eqref{lafina} all the quantities at r.h.s. also depend on the
parameters $(\omega,\nu,\varepsilon)$.

We point out that the nonresonance condition we assume is
\begin{equation}\label{bad}
\begin{aligned}
|\ii \omega \cdot l + \lambda_j^{(\infty)} - \lambda_{j'}^{(\infty)}|
\geq \frac{2\gamma}{\langle l \rangle^\tau \langle j \rangle^\tau
  \langle j'\rangle^\tau}\ ,\quad 
\forall (l , j, j') \neq (0, j, j)\ ,
\end{aligned}
\end{equation}
correspondingly the set of the parameters in which we are able to
prove reducibility is the set of the $(\omega,\nu)$
s.t. \eqref{bad} holds.

\noindent
Finally, in the appendix \ref{sezione appendice}, we collect some
properties on flows of Pseudo-PDEs, Egorov type theorems and norms
that we shall use along our reduction procedure.

\noindent
{\bf Acknowledgments.} Dario Bambusi was supported by GNFM. Riccardo Montalto was supported by the Swiss National Science Foundation, grant {\it Hamiltonian systems of infinite dimension}, project number: 200020--165537. 

\noindent
Part of this
work was done while Riccardo Montalto was visiting Milano with the support of
Universit\`a degli Studi di Milano.

\section{Statement of the main result}\label{sta}

In order to state precisely the main results of the paper, we
introduce some notations.  \\ For any $s \in \R$ we consider the
Sobolev space ${\cal H}^s(\T^d)$ endowed by the norm
$$
\| u\|_{{\cal H}^s} := \Big( \sum_{\xi \in \Z^d} \langle \xi \rangle^{2 s} |\widehat u(\xi)|^2 \Big)^{\frac12} 
$$
where $\langle \xi\rangle := (1 + |\xi|^2)^{\frac12}$ and $\widehat
u(\xi)$ are the Fourier coefficients of $u$. 
Given two Banach spaces $X, Y$ we denote by ${\cal B}(X, Y)$ the space of bounded linear operators $X \to Y$ equipped by the standard operator norm. If $X = Y$, we simply write ${\cal B}(X)$ instead of ${\cal B}(X, X)$.\\
In the following, given $\alpha,\ \beta \in \R,$  we will write $\alpha \lesssim \beta$ if there exists $C>0$ (independent of all the relevant quantities) such that $\alpha \leq C \beta.$ Sometimes we will write $\alpha \lesssim_{s_1, \dots, s_n} \beta$ if $C $ depends on parameters $s_1,\cdots, s_n,$\\
\noindent
We will use the following classes of pseudo-differential operators:
\begin{definition} \label{classe simboli}
	Let $m \in \R$. We say that a ${\cal C}^\infty$ function $a : \T^d \times \R^d \to \C$ is a symbol of class $S^m$ if for any multiindex $\alpha,\ \beta \in \N^d $ there exists a constant $C_{\alpha,\beta} > 0$ such that 
	\begin{equation}
	|\partial_x^\alpha  \partial_\xi^\beta a( x, \xi)| \leq C_{\alpha, \beta} \langle \xi \rangle^{m - |\beta|}\,, \quad \forall (x, \xi) \in \T^d \times \R^d\,. 
	\end{equation}
\end{definition}
\noindent A symbol $a$ defines univocally a linear operator $A$ acting as
$$
A[u] (x) := \sum_{\xi \in \Z^d} a( x, \xi) \widehat u(\xi) e^{\ii x \cdot \xi}\,, \qquad \forall u \in {\cal C}^\infty(\T^d)\,,
$$
that we denote by $A=\Op \big(a\big).$

\begin{definition} \label{simboli non piu' allargati}
	An operator $A$ is called a pseudo-differential operator of order $m$, namely $A\in OPS^{m},$ if there exists $a \in S^{m}$ such that
	$$
	A = Op(a).
	$$
\end{definition}
\noindent The constants $C_{\alpha, \beta}$ of Definition \ref{classe simboli} form a family of seminorms for $S^{m}$ and for $OPS^{m}.$

In the following, we will consider pseudo-differential operators  depending in a smooth way on the angles $\vphi \in \T^{\frak n}$ and in a Lipschitz way on the frequencies $\widetilde \omega = (\omega, \nu) \in \Omega_0 \subseteq \Omega$. We will denote them by $\Lipsp \left(\Om_0; \cinfty \left(\T^{\frak n}; OPS^{m}\right)\right).$ 



We finally state some properties that we will assume to hold on our system \eqref{main equation}:
\begin{definition}[Structural hypotheses]\label{stru}\null
  \goodbreak
\begin{itemize}
		\item[(i)] We say that ${\cal R} \in {\cal B}( L^2(\T^d))$ is a {\it real operator} if it maps real valued functions into real valued functions, namely
		$$u \in L^2(\T^d; \R) \Rightarrow {\cal R}[u] \in L^2(\T^d; \R).$$
		Equivalently, we can say that ${\cal R}$ is a real operator if ${\cal R} = \overline{\cal R}$ where the operator $\overline{\cal R}$ is defined by $\overline{\cal R}[u] := \overline{{\cal R}[\overline u]}$, $u \in L^2(\T^d)$. 
	      \item[(ii)] 	Let $\vphi \mapsto {\cal R}(\vphi), {\cal Q}(\vphi)$ be smooth $\vphi$-dependent families of real operators $\T^{\frak n} \to {\cal B}\big(L^2(\T^d) \big)$; we say that ${\cal R}$ is reversible if 
		\begin{equation}\label{definizione reverisbilita}
		{\cal R} (\vphi) \circ S = - S \circ {\cal R}(- \vphi)\,, \quad \forall \vphi \in \T^{\frak n},
		\end{equation}
		where $S$ is the involution defined by
		\begin{equation}\label{involuzione reversibilita}
		S : L^2(\T^d) \to L^2(\T^d)\,, \quad u(x) \mapsto u(- x).
		\end{equation}
		On the other hand, we say that ${\cal Q}$ is reversibility preserving if 
		\begin{equation}
		{\cal Q} (\vphi) \circ S = S \circ {\cal Q}(- \vphi)\,, \quad \forall \vphi \in \T^{\frak n}.
		\end{equation}
		\item [(iii)]We say that ${\cal R} \in OPS^{1}$ is symmetric hyperbolic if ${\cal R} + {\cal R}^* \in OPS^{0}.$
	\end{itemize}
\end{definition}

We will also consider the case where $V$ is even, namely one has
$$
V(-\varphi,-x)=V(\varphi,x)\ .
$$

 Define the
constant
\begin{equation}\label{definizione s0}
s_0 := \Big[\frac{\frak n}{2 } \Big] + 1\,.
\end{equation}
This paper is devoted to the proof of the following result.

\begin{theorem} \label{main teo}
	Let $V \in {\cal C}^\infty(\T^{\frak n} \times \T^d, \R^d)$,
        ${\cal W} \in \cinfty\left(\T^{\frak n}; OPS^{1
          -\frake}\right)$ and assume that ${\cal W}$ is symmetric
        hyperbolic. Then for any $s \geq s_0$, $\sigma \geq 0$ there
        exists $\e^* >0$ such that $\forall
        \e <\e^*$ there exists a closed set ${\Om_\e \subseteq
          \Om}$ of asymptotically full Lebesgue measure,
        i.e. $\lim_{\e \to 0} |\Omega \setminus \Omega_\e| = 0$, such
        that the following holds: $\forall\ \widetilde \omega = (\om,
        \nu) \in \Om_\e$ there exists a linear bounded and invertible
        operator $ {\cal U}(\vphi) = {\cal U}(\vphi; \widetilde
        \omega) \in {\cal B}({\cal H}^\sigma)$, $\vphi \in \T^{\frak
          n}$ such that, if $u$
        solves \eqref{main equation}, then $v$ defined by $u=\cU(\omega t) v $ solves
\begin{equation}
          \label{Hfin}
	\partialt v = H_\infty v,
	        \end{equation}
	where
        \begin{equation}
H_\infty={\rm diag}(\lambda^{(\infty)}_{j}(\tildeom, \e) )
        \end{equation}
Furthermore, the eigenvalues $\{\lambda^{(\infty)}_{j}(\tildeom,
\e)\}_{j \in \Z^d}$ have the structure
\begin{equation}
  \label{stuct_final}
\lambda^{(\infty)}_{j}(\tildeom, \e)=\ii \nu^{(0)}\cdot
j+z(j)+\cO(\epsilon j^{-m})\ ,
\end{equation}
with $z(.)\in S^{1-\frak e}$ which is also dependent in a Lipschitz way on
$\tilde\omega$, and $\nu^{(0)}=\nu^{(0)}(\tilde \omega) $ which
fulfills
$$
\left|\nu^{(0)}-\nu\right|\leq C\varepsilon\ .
$$
Finally, if the following assumption holds
\begin{itemize}
\item[(Sym)] \label{sym} $V$ is even and ${\cal W}$ is real and reversible ,
\end{itemize} then $ \lambda^{(\infty)}_j \in \ii \R \quad
\forall\ j \in \Z^d.$
\end{theorem}
From the theorem above we can deduce information concerning the
dynamics of the PDE  \eqref{main equation}. 
   \begin{corollary}
    \label{corollario caso rev real}
Under the same assumptions of Theorem \ref{main teo}, but not (Sym)
only one of the following two possibilities occurs
\begin{itemize}
\item[(1)] All the solutions of \eqref{main equation} are almost
  periodic and
  \begin{equation}
    \label{stisob}
u_0 \in {\cal H}^\sigma\ \Longrightarrow\ \| u(t, \cdot) \|_{{\cal H}^\sigma} \lesssim \| u_0\|_{{\cal H}^\sigma}
    \end{equation}
  uniformly w.r. to $t \in\R$.
\item[(2)] There exist $a,C>0$ and  some initial data $u_0$ s.t.
  \begin{equation}
    \label{stisob.2}
  \| u(t, \cdot) \|_{{\cal H}^\sigma} \geq Ce^{ a|t|} \| u_0\|_{{\cal H}^\sigma}
  \end{equation}
either for $t>0$ or for $t<0$ or for $t\in\R$.
\end{itemize}
  \end{corollary}
We remark that under the assumption (Sym) only possibility (1) occurs.

%
%
%



\section{Regularization up to smoothing remainders}\label{sezione regolarizzazione grande}

In this section we conjugate the vector field
\begin{equation}\label{campo vettoriale trasporto}
H(\vphi) := \big( \nu + \e V(\vphi, x) \big) \cdot \nabla + \e {\cal W}^{(0)}(\vphi)\,, \quad {\cal W} \in OPS^{1 - \frak e}
\end{equation}
to another one which is a smoothing perturbation of a time independent
diagonal operator.

\noindent
First remark that a time dependent linear invertible
  transformation $u=\Phi(\omega t)u'$ transforms the equation $\dot
  u=Hu$ into the equation $\dot u'=H'u'$, where
  $$
H'=\Phi_{\omega*}H:=\Phi(\vphi)^{-1}[H\Phi(\vphi)-\omega\cdot
  \partial_\vphi\Phi(\vphi)]\ .
  $$

\begin{definition}[Lipschitz norm]\label{norma generale Lip gamma}
Given a Banach space $(X, \| \cdot \|_X)$, a set $\Omega_0 \subset \Omega = [1, 2]^{\frak n + d}$, $\gamma > 0$ and a Lipschitz function $f : \Omega_0 \to X$, we denote by $\| \cdot \|^{\Lipg}_X$ the Lipschitz norm defined by
\begin{equation}\label{norma generale Lip gamma}
\begin{aligned}
& \|f \|^{\Lipg}_X := \|f \|^{\rm sup}_X + \gamma \|f \|^{\rm lip}_X\,, \\
& \|f \|^{\rm sup} := \sup_{\widetilde \omega \in \Omega_0} \|f(\widetilde \omega) \|_X\,, \quad \|f \|^{\rm lip}_X := \sup_{\begin{subarray}{c}
\widetilde \omega_1, \widetilde \omega_2 \in \Omega_0 \\
\widetilde \omega_1 \neq \widetilde \omega_2
\end{subarray}} \dfrac{\|f(\widetilde \omega_1) - f(\widetilde \omega_2)\|_X}{|\widetilde \omega_1 - \widetilde \omega_2|}\,. 
\end{aligned}
\end{equation}
In the case where $\gamma = 1$, we simply write $\| \cdot \|_X^{\rm
  Lip}$ for $\| \cdot\|_X^{ \rm Lip(1)}$. If $X = \C$ we write $|
\cdot |^\Lipg, |\cdot|^{\rm sup}, |\cdot |^{\rm lip}$ for $\| \cdot
\|_\C^{\Lipg}, \| \cdot \|_\C^{\rm sup}, \| \cdot \|_\C^{\rm
  lip}$. 
\end{definition}

\subsection{Reduction to constant coefficients of the highest order term}\label{sezione riduzione ordine alto}

We consider a diffeomorphism of the torus $\T^{d}$ of the form 
$$
\T^{d} \to \T^{d}, \quad x \mapsto x +  \alpha(\vphi, x)
$$
where $\alpha \in {\cal C}^\infty(\T^{\frak n} \times \T^d, \R^d)$ is a function to be determined. It is well known that for $\| \alpha \|_{C^1}$ small enough such a diffeomorphism is invertible and its inverse has the form 
$$
\T^{d} \to \T^{d}, \quad y \mapsto y + \widetilde \alpha(\vphi, y)
$$
with $\widetilde \alpha \in {\cal C}^\infty(\T^{\frak n} \times \T^d, \R^d)$. We then consider the transformation 
\begin{equation}\label{definizione cal A}
{\cal A}(\vphi) : u(x) \mapsto u(x + \alpha(\vphi, x))\,, \quad \vphi \in \T^{\frak n}
\end{equation}
whose inverse is given by 
\begin{equation}\label{definizione cal A inverso}
{\cal A}(\vphi)^{- 1} : u(y) \mapsto u(y + \widetilde \alpha(\vphi, y))\,, \quad \vphi \in \T^{\frak n}\,. 
\end{equation}
A direct calculation shows that the {\it quasi-periodic push-forward} of the vector field $H(\vphi)$ is given by 
\begin{align}
H^{(0)}(\vphi)  = {\cal A}_{\omega *} H(\vphi) = V^{(0)}(\vphi, x) \cdot \nabla + \e {\cal W}^{(0)}(\vphi) \label{H0 ordine alto}
\end{align}
where 
\begin{equation}\label{definizioni V (0) cal W (0)}
\begin{aligned}
V^{(0)}(\vphi, x) & :={\cal A}(\vphi)^{- 1}\Big( \omega \cdot \partial_\vphi \alpha + \nu +\e V +  \big( \nu + \e V \big) \cdot \nabla \alpha \Big) \\
{\cal W}^{(0)}(\vphi) & := {\cal A}(\vphi)^{- 1} {\cal W}(\vphi) {\cal A}(\vphi)\,. 
\end{aligned}
\end{equation}
The following proposition is a direct consequence of Proposition 3.4
in \cite{trasporto paper} to which we refer for the proof. It allows
to choose the function $\alpha(\vphi, x)$ so that the highest order
term $V^{(0)}(\vphi , x) \cdot \nabla$ in \eqref{H0 ordine alto} is
reduced to constant coefficients.
\begin{proposition}\label{esempio1}
Let $\gamma \in (0, 1)$ and $\tau > \frak n + d$. There exists a
Lipschitz  function $\muzero : \Omega \to \R^{d}, {\tildeom} \mapsto
{\muzero(\tildeom)}$ (where we recall that $\Omega := [1, 2]^{\frak n + d}$) such that
\begin{equation}\label{tordo4b}
\qquad \lvert  \muzero ({\tildeom})- \nu  \rvert^{\Lipg}\lesssim \e , \,
\end{equation}
and, in the set
\begin{equation}\label{cantor regolarizzazione}
\Omega_{0, \gamma}:= \Big\{{\tildeom} \in \Omega: \; |\omega \cdot l
+\muzero ({\tildeom})\cdot j  |>\frac{\g }{\langle l
  ,j\rangle^{\tau}}\,,\;\forall (l ,j)\in \Z^{\nu+d} \setminus
\{0\}\Big\} \ ,
\end{equation}
the following holds. There exists a map
\begin{equation}\label{perunpugnodidollarib}
\alpha:  \T^{\nu+d}\times \Omega_{0, \gamma}\to \R^d \,,  
\end{equation}
so that the map $\T^{\frak n + d} \to \T^{\frak n + d}$, $(\vphi, x)
\mapsto (\vphi, x + \alpha(\vphi, x))$ is a diffeomorphism with
inverse given by $(\vphi, y) \mapsto (\vphi, y + \widetilde
\alpha(\vphi, y))$, furthermore
\begin{equation}\label{stima widetilde alpha}
\|\alpha \|^{\Lipg}_{s} \lesssim_s  \e \g^{-1},\quad \| \widetilde \alpha\|^{\Lipg}_s \lesssim_s \e \gamma^{- 1}\,, \quad \forall s \geq 0\,. 
\end{equation}
Moreover for any ${\tildeom} \in \Omega_{0, \gamma}$ $V^{(0)}$
reduces to a constant (as a function of $x$ and $\varphi$), namely
\begin{equation}\label{tordo6b}
\begin{aligned}
{V^{(0)}} ={\cal A}^{- 1}(\vphi)\Big( \omega \cdot \partial_\vphi \alpha + \nu +\e V +  \big( \nu + \e V \big) \cdot \nabla \alpha \Big)=  \muzero ({\tildeom}). 
\end{aligned}
\end{equation}
Finally, if $V$ is even, then $\alpha$ and $\widetilde \alpha$ are odd.
\end{proposition}
\begin{remark}\label{misura Omega 0 gamma}
By standard arguments one has $|\Omega \setminus \Omega_{0, \gamma}|
\lesssim \gamma$. More precisely, on the one side one has that
  vectors which are Diophantine with constant $\gamma$ have complement with
   measure of order $\gamma$, and on the other, Lipschitz maps
   preserve the order of magnitude of the measure of sets.
\end{remark}
\begin{remark}\label{stime trasformazione cal A}
Using the definitions \eqref{definizione cal A}, \eqref{definizione cal A inverso} and the estimates \eqref{perunpugnodidollarib}, \eqref{stima widetilde alpha}, a direct calculation shows that the map $\T^{\frak n} \mapsto {\cal B}({\cal H}^s)$, $\vphi \mapsto {\cal A}(\vphi)^{\pm 1}$ is bounded for any $s \geq 0$ and 
$$
\begin{aligned}
& \sup_{\vphi \in \T^{\frak n}}\|{\cal A}(\vphi)^{\pm 1} - {\rm Id} \|_{{\cal B}({\cal H}^{s + 1}, {\cal H}^s)} \lesssim_s \e \gamma^{- 1}, \quad \forall s \geq 0\,, \\
& \sup_{\vphi \in \T^{\frak n}} \| \partial_\vphi^\alpha {\cal A}(\vphi)^{\pm 1} \|_{{\cal B}({\cal H}^{s + |\alpha|}, {\cal H}^s)} \lesssim_{s, \alpha } \e \gamma^{- 1}, \quad \forall s \geq 0, \quad \forall \alpha \in \N^{\frak n}\,. 
\end{aligned}
$$
\end{remark}
Recalling \eqref{H0 ordine alto}, \eqref{definizioni V (0) cal W (0)}
and applying Proposition \ref{esempio1} one gets that the vector field
$H^{(0)}(\vphi)$ takes the form
\begin{equation}\label{ordine alto coefficienti costanti}
H^{(0)}(\vphi) = \muzero  \cdot \nabla + \e {\cal W}^{(0)}(\vphi)
\end{equation}

We now study the properties of $\cW^{(0)}$.

\begin{lemma}\label{lemma cal W (0) simbolo}
One has that ${\cal W}^{(0)} \in {\cal L}ip\Big( \Omega_{0, \gamma},
{\cal C}^\infty\Big( \T^{\frak n}, OPS^{1 - \frak e}\Big)
\Big)$. Moreover ${\cal W}^{(0)}$ is symmetric
hyperbolic. Furthermore, if $V$ is even and $\cW$ real and reversible, then
${\cal W}^{(0)}$ is real and reversible.
\end{lemma}
\begin{proof}
Let $\Phi(\vphi) := {\cal A}(\vphi)^{- 1}$, i.e. $\Phi(\vphi)[u](y) = u (y + \widetilde \alpha(\vphi, y))$ and for any $\tau \in [0, 1]$ we consider $\Phi(\tau, \vphi)[u](y) := u(y + \tau \widetilde \alpha(\vphi, y))$. Let $\psi(\tau, \vphi, y) := \Phi(\tau, \vphi)[u](y)$, then $\psi(0, \vphi, y) = u(y)$ and 
\begin{equation}\label{uffa uffa 0}
\partial_\tau \psi = a(\tau, \vphi, y) \cdot \nabla \psi\,, \quad a(\tau, \vphi, y) := \big({\rm Id} + \tau \nabla \widetilde \alpha(\vphi, y) \big)^{- 1} \widetilde \alpha(\vphi, y)\,.
\end{equation}
Then by the Egorov theorem (see Theorem A.0.9 in \cite{Taylor}) it follows that ${\cal W}^{(0)} \in {\cal L}ip\Big( \Omega_{0, \gamma},
{\cal C}^\infty\Big( \T^{\frak n}, OPS^{1 - \frak e}\Big)
\Big)$.

\noindent
We now show that ${\cal W}^{(0)}$ is symmetric
hyperbolic. Since by \eqref{perunpugnodidollarib}, \eqref{stima widetilde alpha} the functions $\alpha, \widetilde \alpha = O(\e \gamma^{- 1})$ one has that 
$$
{\rm det}\big( {\rm Id} + \nabla \alpha \big)\,,\, {\rm det}\big( {\rm Id} + \nabla \widetilde \alpha \big) > 0
$$
for $\e \gamma^{- 1}$ small enough. Moreover, using that $y \mapsto y + \widetilde \alpha(y)$ is the inverse diffeomorphism of $x \mapsto x + \alpha(x)$ one gets that 
\begin{equation}\label{che palle 2}
{\rm det}\big({\rm Id} + \nabla \widetilde \alpha(y) \big) = \dfrac{1}{{\rm det}\big({\rm Id} + \nabla \alpha \big)|_{x = y + \widetilde \alpha(y)}}\,. 
\end{equation}
A direct calculation shows that 
$$
{\cal A}^* = {\rm det}\big( {\rm Id} + \nabla \widetilde \alpha \big) {\cal A}^{- 1}\,, \quad ({\cal A}^{- 1})^* = {\rm det}\big( {\rm Id} + \nabla \alpha \big) {\cal A}\,.
$$
Then 
\begin{align}
({\cal W}^{(0)})^* & = ({\cal A}^{- 1} {\cal W} {\cal A})^* = {\cal A}^* {\cal W}^* ({\cal A}^{- 1})^* \nonumber\\
& = {\rm det}\big( {\rm Id} + \nabla \widetilde \alpha \big) {\cal A}^{- 1} {\cal W}^* {\rm det}\big( {\rm Id} + \nabla \alpha \big) {\cal A} \nonumber\\
& = {\rm det}\big( {\rm Id} + \nabla \widetilde \alpha \big) {\cal A}^{- 1} {\rm det}\big( {\rm Id} + \nabla \alpha \big)  {\cal W}^* {\cal A} \nonumber\\
& \quad + {\rm det}\big( {\rm Id} + \nabla \widetilde \alpha \big) {\cal A}^{- 1} [{\cal W}^*\,,\, {\rm det}\big( {\rm Id} + \nabla \alpha \big)] {\cal A}\,.  \label{che palle 0}
\end{align}
Since ${\cal W}^* \in OPS^{1 - \frak e}$ one has that the commutator $[{\cal W}^*\,,\, {\rm det}\big( {\rm Id} + \nabla \alpha \big)] \in OPS^{- \frak e} \subset OPS^0$. Using that ${\cal A}(\vphi)^{- 1} = \Phi(\vphi)$ is the time 1 flow map of the PDE \eqref{uffa uffa 0}, by applying the Egorov Theorem A.0.9 in \cite{Taylor}, one gets that ${\rm det}\big( {\rm Id} + \nabla \widetilde \alpha \big) {\cal A}^{- 1} [{\cal W}^*\,,\, {\rm det}\big( {\rm Id} + \nabla \alpha \big)] {\cal A} \in OPS^0$. hence  
\begin{align}
({\cal W}^{(0)})^* & = {\rm det}\big( {\rm Id} + \nabla \widetilde \alpha \big) {\cal A}^{- 1} {\rm det}\big( {\rm Id} + \nabla \alpha \big)  {\cal W}^* {\cal A}  + OPS^0 \nonumber\\
& = {\rm det}\big( {\rm Id} + \nabla \widetilde \alpha \big)  {\rm det}\big( {\rm Id} + \nabla \alpha \big)|_{x = y + \widetilde \alpha(y)} {\cal A}^{- 1}  {\cal W}^* {\cal A}  + OPS^0   \nonumber\\
&\stackrel{\eqref{che palle 2}}{=}{\cal A}^{- 1} {\cal W}^* {\cal A} + OPS^0\,. \label{che palle 3}
\end{align}
Finally, using that ${\cal W}$ is symmetric hyperbolic, i.e. ${\cal W} + {\cal W}^* \in OPS^0$, by \eqref{che palle 2} and applying again the Egorov Theorem A.0.9 in \cite{Taylor} to deduce that ${\cal A}^{- 1}({\cal W} + {\cal W}^*){\cal A} \in OPS^0$ one gets that ${\cal W}^{(0)} + ({\cal W}^{(0)})^* \in OPS^0$. In the real and reversible case, one has that ${\cal W}$ is a reversible operator. By Proposition \ref{esempio1}, one has that $\alpha, \widetilde \alpha$ are odd functions, implying that ${\cal A}$, ${\cal A}^{- 1}$ are reversibility preserving operators. Hence one concludes that ${\cal W}^{(0)} = {\cal A}^{- 1} {\cal W} {\cal A}$ is a reversible operator. 
\end{proof}

\subsection{Reduction of the lower order terms}\label{sezione regolarizzazione}

\noindent The reduction of the lower order terms is contained in the following result, which is an adaptation of Theorem 3.8 of \cite{BGMR2} to a
symmetric hyperbolic context.

\begin{theorem} \label{regula}
$\forall\ M>0$ there exists a sequence of symmetric hyperbolic maps
  $\{{ G}_j(\vphi, \tildeom)\}_{j=1}^{M}$ with ${ G}_j(\vphi,
  \tildeom) \in {\Lipsp} \left( \Omegauno; \cinfty\left( \T^{{\frak
      n}}; OPS^{1-j\frake}\right)\right)$ such that the change of
  variables $\psi = e^{- \e G_1(\vphi, \tildeom)} \cdots \ e^{- \e
    G_M(\vphi, \tildeom)}\phi$ transforms $H_0 + \e {\cal
    W}^{(0)}(\vphi)$ into the operator
	\begin{equation} \label{regolare}
		H^{(M)}(\vphi) = H_0 + \e Z^{(M)}(\tildeom) + \e {\cal W}^{(M)}(\vphi, \tildeom),
	\end{equation}
	where $Z^{(M)}$
	is a time independent Fourier multiplier, which in particular
        fulfills
	\begin{equation} \label{z diag}
	[Z^{(M)}, K_m] = 0, \quad m= 1 \dots, d,
	\end{equation}
        and
        \begin{equation} \label{regolarita z w}
	\begin{gathered}
	Z^{(M)}(\tildeom) \in \Lipsp \left(\Omegauno; OPS^{1-\frake}\right),\\
	{\cal W}^{(M)}(\vphi, \tildeom) \in \Lipsp \left(\Omegauno; {\cal C}^{\infty}(\T^{{\frak n}}; OPS^{1-M\frake})\right).
	\end{gathered}
	\end{equation}
	Furthermore, 	if ${\cal W}^{(0)}$ is real and reversible, then $Z^{(M)},\ {\cal W}^{(M)}$ are real and reversible too.
\end{theorem}

We now prove such theorem.

Denote ${K_j = i \partial_j, \quad j = 1, \dots, d,}$ then $K_1,
\dots, K_d$ are self-adjoint commuting operators such that $K_m \in
OPS^{1}\ \forall m= 1, \dots, d.$ Define $K=\left(K_1, \dots,
K_d\right).$ 
\noindent
The main step for the proof of Theorem \ref{regula} is the following
lemma, which is a variant of Lemma 3.7 of \cite{BGMR2}:
\begin{lemma} \label{lemma eq homolog}
	Let $ {W} \in \Lipsp \left(\Omegauno; \cinfty\left(\T^{{\frak
            n}}; OPS^{\eta}\right)\right),$ be given and consider the
        homological equation
	\begin{equation} \label{homological equation reg}
	\omphi G + [H_0,\ G] = W - \langle W\rangle 
	\end{equation}
	with 
	$$
	\langle W \rangle  := \frac{1}{(2\pi)^{{\frak
              n}+d}}\int_{\T^{d}} \int_{\T^{{\frak n}}} e^{i \tau
          \cdot K} W e^{ -i \tau \cdot K}\ d\vphi \ d \tau\ ;
	$$
then \eqref{homological equation reg} has a solution $G \in \Lipsp \left(\Omegauno; \cinfty\left(\T^{{\frak n}}; OPS^{\eta}\right) \right). $\\
	If $W$ is symmetric hyperbolic, $G$ is symmetric hyperbolic.
	Moreover, if $W$ is real and reversible, $G$ is real and reversibility preserving; if $W$ is anti self-adjoint, $G$ is anti self-adjoint.
\end{lemma}
\begin{proof} Define $\forall \tau \in \T^{d}$
	$$
W(\tau) := e^{i \tau \cdot K} W e^{- i\tau \cdot K},
	$$
then we look for $G$ s.t.
$$
	G(\tau) := e^{i \tau \cdot K} G e^{ -i \tau \cdot K}
$$ solves
	\begin{equation} \label{eq omologica allargata}
	\omphi G(\tau) +  [H_0,\ G(\tau)] = W(\tau) - \langle W \rangle \quad \forall\ \tau \in \T^d,
	\end{equation}
	observing that  since $G= G(0),\ W= W(0),$ solving equation
        \eqref{eq omologica allargata} $\forall\ \tau$ implies having
        solved~ \eqref{homological equation reg}.\\
Note that
$\forall\ \eta \in \R,\ \forall\ A \in OPS^{\eta}$ the map
		\begin{equation} 
		[-1,\ 1] \ni \tau \mapsto e^{-i \tau \cdot K} A e^{ i \tau \cdot K} \in \cinfty \left(\T^d; OPS^{\eta}\right)
		\end{equation}
(see Remark \ref{lemma assumption ii} of the appendix).
	We make a Fourier expansion both in $\vphi$ and $\tau$ variables, namely
	\begin{equation} \label{fourier op}
	W_{\tildeom}( \vphi, \tau) =  \sum_{k \in \Z^{d}} \sum_{l \in \Z^{{\frak n}}} \four{W}{k}{l}(\tildeom) e^{i \vphi \cdot l} e^{i \tau \cdot k},
	\end{equation}
        and similarly for $G$. 
	A direct calculation shows that 
	\begin{align*}
	[H_0,\ G(\tau)]  = \sum_{k,\ l} i \left(\muzero  \cdot k\right)\ \four{G}{k}{l}e^{i\tau \cdot k} e^{i \vphi \cdot l}\,. 
	\end{align*}
	Thus, taking the $(k,l)-$th Fourier coefficient of equation \eqref{eq omologica allargata}, one has
	$$
	i \left( \om \cdot l +  \muzero \cdot k \right) \four{G}{k}{l} = \four{W}{k}{l} \quad \textrm{ if }(k, l) \neq (0,0), \quad \four{G}{0}{0} = 0.
	$$
For $\left|k\right|+\left| l\right|\not=0$,        define
        $$
\four{G}{k}{l}:=\frac{\four{W}{k}{l}}{i ( \om \cdot l +  \muzero \cdot
  k)}\ ,
        $$ then, by regularity of the map $(\vphi, \tau) \mapsto
W(\vphi, \tau)$ all the seminorms of the operator $\four{W}{k}{l}$
decay faster than any power of $(|k|+|l|),$ and since the frequencies
belong to $\Omega_{0,\gamma}$ (cf. \eqref{cantor regolarizzazione}), it follows that
the seminorms of the operator $\four{G}{k}{l}$ exhibit the same decay;
hence  the series defining $G(\tau)$ converges absolutely and  $G
= G(0) \in \cinfty\left(\T^{{\frak n}};
OPS^{\eta}\right).$
\\ Lipschitz regularity with respect to
$\tildeom=(\om, \nu) \in \Om_{0, \gamma }$ follows observing that
given $(\omega_1, \nu_1), (\omega_2, \nu_2) \in \Omega_{0, \gamma}$,
one has that
	$$
	\begin{aligned}
	  \four{G}{k}{l}(\omega_1) -  \four{G}{k}{l}(\omega_2)   & =  \four{G}{k}{l}(\omega_1) \dfrac{(\omega_1 - \omega_2) \cdot l  + \big( \muzero(\omega_1, \nu_1) - \muzero (\omega_2, \nu_2) \big) \cdot k}{(\omega_1 \cdot l + \muzero(\omega_1, \nu_1) \cdot k) (\omega_2 \cdot l + \muzero(\omega_2, \nu_2) \cdot k)} \\
	  & \quad + \dfrac{\four{G}{k}{l}(\omega_1) - \four{G}{k}{l}(\omega_2)}{\ii (\omega_2 \cdot l + \muzero (\omega_2, \nu_2) \cdot k)}
	\end{aligned}
	$$
	using the fact that the map $(\omega, \nu) \mapsto \muzero (\omega, \nu)$ is Lipschitz (see Proposition \ref{esempio1}) and the diophantine estimate required in \eqref{cantor regolarizzazione}.\\
	{\sc Symmetric hyperbolicity:} We observe that
	$$
	W + W^{*} = e^{- i \tau \cdot K} \left( W(\tau) + W^{*}(\tau)\right) e^{ i \tau \cdot K}, \quad G + G^{*} = e^{- i \tau \cdot K} \left( G(\tau) + G^{*}(\tau)\right) e^{ i \tau \cdot K}.
	$$
	Hence $W$ (resp., $G$) is symmetric hyperbolic if and only if $W(\tau)$ (resp., $G(\tau)$) is symmetric hyperbolic.\\
	Thus, arguing as before and being
	$$
	\widehat{\left(W^{*}\right)}_{k,l} = \overline{\widehat{W}_{-k,-l}}\quad \forall\ k \in \Z^d,\ l \in \Z^{\frak n},
	$$
	it follows that if $\forall\ k \in \Z^d,\ l \in \Z^{\frak n}$
	$
	\widehat{W}_{k,l} + \overline{\widehat{W}_{-k,-l}}
	$
	are the Fourier coefficients of an operator in $OPS^{0},$ then
	$$
	\widehat{G}_{k,l} + \overline{\widehat{G}_{-k,-l}} = \frac{\widehat{W}_{k,l} + \overline{\widehat{W}_{-k,-l}}}{i \left( \om \cdot l +  \nu \cdot k \right)}
	$$
	are again Fourier coefficients of an operator in $OPS^{0}.$ \\
	{\sc Reversibility:}
	We apply Lemma \ref{lemma tau rev} of the Appendix to deduce
        reversibility of $W$ and we observe that an operator $A(\tau,
        \vphi)$ is reversible (resp. reversibility preserving) if and
        only if, developing in Fourier series as in \eqref{fourier op}, its coefficients satisfy
	$$
	\four{A}{k}{l} \circ S = - S \circ \four{A}{-k}{-l} \qquad \left(\textrm{resp.\ }\four{A}{k}{l} \circ S =  S \circ \four{A}{-k}{-l}\right),
	$$
	so that $\forall\ k \in \Z^d,\ l \in \Z^{\frak n},$
	\begin{align*}
	\four{G}{k}{l} \circ S &= \frac{\four{W}{k}{l} \circ S}{i (\omega \cdot l + \nu \cdot k)}=  \frac{-S \circ \four{W}{-k}{-l}}{- (\omega \cdot (-l) + \nu \cdot (-k))} = S \circ \four{G}{-k}{-l}.
	\end{align*}
	Hence $G,$ and thus ${\calg},$ is reversibility preserving. (See Lemma \ref{lemma tau rev}.)\\
	{\sc Reality:}
	Reality condition in Fourier coefficients reads
	$$
	\four{A}{l}{k} = \overline{\four{A}{-l}{-k}}.
	$$
	We apply Lemma \ref{lemma tau rev} again to deduce that reality of $W$ (resp, $G$) is equivalent to reality of $\calw$ (resp, $\calg$) and we compute
	\begin{align*}
	\four{G}{k}{l} &= \frac{\four{W}{k}{l}}{i (\omega \cdot l + \nu \cdot k)} =\frac{\overline{\four{W}{-k}{-l}}}{-i (\omega \cdot (-l) + \nu \cdot (-k))}= \overline{\four{G}{-k}{-l}}.
	\end{align*}
\end{proof}

\begin{proof}[Proof of Theorem \ref{regula}]
	Fix $M>0.$ We prove by induction that $\forall j = 0, \dots, N-1$
	$$
	H^{(j)}(\vphi) = H_0 + \e Z^{(j)}(\tildeom) + \e {\cal W}^{(j)}(\vphi, \tildeom)
	$$ is mapped by the change of variables 
	\begin{equation} \label{change coord reg}
	u =e^{- \e G_j(\vphi, \tildeom)} v
	\end{equation}
	into
	$$
	H^{(j+1)}(\vphi) = H_0 + \e Z^{(j+1)}(\tildeom) + \e {\cal W}^{(j+1)}(\vphi, \tildeom),
	$$
	with
	\begin{equation} \label{regolarita z w}
	\begin{gathered}
	Z^{(j+1)}(\tildeom) \in  {\cal L}ip\Big( \Omega_{0, \gamma}; {\cal C}^{\infty}(\T^{{\frak n}}; OPS^{1-\frake}) \Big),\\ {\cal W}^{(j+1)}  \in \Lipsp \left(\Omegauno ;   {\cal C}^{\infty}\left(\T^{{\frak n}}; OPS^{1-(j+1)\frake}\right)\right),
	\end{gathered}
	\end{equation}
	${\cal W}^{(j+1)}$ symmetric hyperbolic and $Z^{(j+1)}(\tildeom)$ a Fourier multiplier commuting with all the $K_m.$\\
	If $j=0,$ the hypotheses are satisfied for $Z^{(0)}=0,\ {\cal W}^{(0)}=\calw \in {\Lipsp \left(\Omegauno; \cinfty(\T^{{\frak n}}; OPS^{1-\frake})\right)}.$\\
	Suppose now that $H^{(j)}$ satisfies the required hypotheses; the change of coordinates \eqref{change coord reg} maps $H^{(j)}$ into
	\begin{align}
	H^{(j+1)}(\vphi, \tildeom) & = H_0 +  \e Z^{(j)}(\tildeom) + \e \langle {\cal W}^{(j)} \rangle\\ \label{1st order term}
	&+ \e \left(-  \omphi G_j +  [H_0, G_j] +  {\cal W}^{(j)}(\vphi, \tildeom) -   \langle {\cal W}^{(j)} \rangle\right)\\ \label{resto ho}
	&+e^{ \e G_j(\vphi, \tildeom)}H_0 e^{ -\e G_j(\vphi, \tildeom)} - H_0 -\e [H_0, G_j]\\ \label{resto z}
	&+\e e^{ \e G_j(\vphi, \tildeom)} Z^{(j)}(\tildeom) e^{-\e G_j(\vphi, \tildeom)} - \e Z^{(j)}(\tildeom)\\ \label{resto v}
	&+\e e^{\e G_j(\vphi, \tildeom)} {\cal W}^{(j)}(\vphi, \tildeom) e^{-\e G_j(\vphi, \tildeom)} - \e {\cal W}^{(j)}(\vphi, \tildeom)\\  \label{resto xpunto}
	&- \e \int_{0}^{1}\ e^{-\e  s G_j(\vphi, \tildeom)} \omphi G_j(\vphi, \tildeom) e^{ \e s G_j(\vphi, \tildeom)} \ ds + \e \omphi G_j.
	\end{align}

	From Lemma \ref{lemma eq homolog} it is possible to find an operator $G_j \in \Lipsp \left(\Omegauno; \cinfty(\T^{{\frak n}}; OPS^{1-j\frake})\right)$ such that $G_j$ is symmetric hyperbolic and \eqref{1st order term} equals zero.
	Since Lemma \ref{lemma coniugo col flusso} of the Appendix entails that
	\begin{gather*}
	\eqref{resto ho} \in  \Lipsp\left(\Omegauno; \cinfty\left(\T^{{\frak n}}; OPS^{1-2j\frake}\right)\right),\\
	\eqref{resto z} \in  \Lipsp\left(\Omegauno; \cinfty\left(\T^{{\frak n}}; OPS^{1-(j+1)\frake}\right)\right),\\
	\eqref{resto v} \in  \Lipsp\left(\Omegauno; \cinfty\left(\T^{{\frak n}}; OPS^{1-2j\frake}\right)\right),\\
	\eqref{resto xpunto}  \in  \Lipsp\left(\Omegauno; \cinfty\left(\T^{{\frak n}}; OPS^{1-2j\frake}\right)\right),
	\end{gather*}
	if we define
	\begin{equation}
	\begin{aligned}
 Z^{(j+1)}(\tildeom) & : =  Z^{(j)}(\tildeom) +  \langle {\cal W}^{(j)} \rangle, \quad \\
	\e {\cal W}^{(j+1)}(\vphi, \tildeom) & = \eqref{resto ho} + \eqref{resto z} + \eqref{resto v} + \eqref{resto xpunto},
	\end{aligned}
	\end{equation}
	we have ${\cal W}^{(j+1)}(\vphi, \tildeom) \in \Lipsp\left(\Omegauno; \cinfty\left(\T^{{\frak n}}; OPS^{1-(j+1)\frake}\right)\right).$\\
	We observe that \eqref{resto ho} is of order $\e,$ as can be seen performing a Taylor expansion of the operator $e^{- \e G_j(\vphi, \tildeom)}H_0 e^{ \e G_j(\vphi, \tildeom)}$ as in Lemma \ref{lemma coniugo col flusso} of the Appendix.\\
	Reality and reversibility of ${\cal W}^{(j+1)}(\vphi, \tildeom)$ follow from  Lemma \ref{lemma flusso PDE}, whereas symmetric hyperbolicity of ${\cal W}^{(j+1)}(\vphi, \tildeom)$ follows from Lemma \ref{lemma commutatori adj}.
\end{proof}

\begin{remark} \label{rmk regolarita exp g}
	For all $j = 1,\ \dots,\ M$ we have $e^{\e G_j }\in {\cal B}\left({\cal H}^{\s}\right)\ \forall\ \s, $ and
	$$
	\|e^{\e G_j } - {\rm Id} \| \pediceSunoSdue{\s}{\s - (1-j\frake)} \lesssim \e \| G_j\| \pediceSunoSdue{\s}{\s - (1-j\frake)}.
	$$
	Furthermore, from Lemma \ref{lemma flusso PDE}, $\forall\ \alpha \in \N$ we have
	$$
	\partial_{\vphi}^{\alpha} e^{\e G_j } \in \bSunoSdue{\s}{\s - (1-j\frake)|\alpha|}.
	$$
\end{remark}
}
Note that, since $Z^{(M)} \in {\cal L}ip\Big( \Omega_{0, \gamma};
{\cal C}^\infty(\T^{\frak n}; OPS^{1 - \frak e}) \Big)$ then $Z^{(M)}
= {\rm Op}(z(\xi))$ with $z \in {\cal L}ip\Big( \Omega_{0, \gamma};
{\cal C}^\infty(\T^{\frak n}; S^{1 - \frak e}) \Big)$. Hence
$\partial_\xi z \in {\cal L}ip\Big( \Omega_{0, \gamma}; {\cal
  C}^\infty(\T^{\frak n}; S^{ - \frak e}) \Big)$ and the following
estimate holds
\begin{gather} \label{z lip 1}
\sup_{\xi \in \R^d} \langle \xi \rangle^{\frak e - 1} |z|^{\rm Lip}\,,\, \underset{ \xi \in \R^d}{\sup}  \langle \xi \rangle^{1-\frak e }|\partial_{\csi}z(\xi, \cdot)|^{\rm Lip} \lesssim \e;
\end{gather}
Concerning the second of \eqref{z lip 1}, we remark that we will only
use the fact that $|\partial_{\csi}z(\xi, \cdot)|^{\rm Lip}$ is
bounded.

\section{Reducibility}\label{sezione riducibilita grande}
\subsection{Functional Setting}
Given a linear operator ${R} : L^2(\T^d) \to L^2(\T^d),$ we denote by $R_j^{j'}$ its matrix elements with respect to the exponential basis $\{ e^{\ii j \cdot x} : j \in \Z^d \},$ namely 
$$
{ R}_j^{j'} := \int_{\T^d} { R}[e^{\ii j'\cdot x}]\,\, e^{- \ii j \cdot x}\, d x\,, \quad \forall j, j' \in \Z^d\,. 
$$ 
We define some families of operators related to ${ R} \in {\cal B}(L^2(\T^d))$ that will be useful in our estimates:
\begin{definition}
	Given $\beta \geq 0$ and ${ R} \in {\cal B}(L^2(\T^d)),$ we define the operator $\langle \nabla \rangle^{\beta} { R}$ as
	$$
	(\Gradbeta { R})_{j}^{j'} := \langle j - j' \rangle^\beta { R}_j^{j'}.
	$$ 
\end{definition}
We remark that this operator is useful since, for any operator
  $R$ and any function $u$, one has
  $$
\nabla R u=R\nabla u+[R;\nabla ] u \ , 
$$
and
$$
[R;\nabla ]\simeq \langle\nabla\rangle R\ .
$$

\begin{definition}
	We consider the space
	$$
	{\cal B}^{HS}({\cal H}^{\s_1}, {\cal H}^{\s_2}) := \left \lbrace { R} \in \bSunoSdue{\sigma_1}{\sigma_2} \big|\ \|{ R}\|^{HS}_{\s_1, \s_2} < + \infty \right \rbrace,
	$$
	with
	$$
	\left(\| { R} \|^{HS}_{\sigma_1, \sigma_2}\right)^2 := \sum_{k \in \Z^d} \sum_{k' \in \Z^d} \langle k \rangle ^{2 \sigma_2} |{ R}_{k}^{k'}|^2 \langle k' \rangle ^{-2 \sigma_1}.
	$$
\end{definition}
We consider operators ${ R}(\vphi)$ depending on the angles $\vphi \in \T^{\frak n},$ with ${ R} \in {\cal H}^{s}\left(\T^{\frak n};\ \bhsSunoSdue{\sigma_1}{\sigma_2} \right).$ Thus we define the time Fourier coefficients of ${ R}:$  $\forall\ l \in \Z^{\frak n}$ 
$\hat{R}(l)$ is the operator with matrix elements
\begin{equation}\label{def fourier l}
(\hat{ R}(l))_{j}^{j'} := \frac{1}{(2\pi)^{{\frak n}}} \int_{\T^{\frak n}} { R}_{j}^{j'} e^{-\ii l \vphi}\ \ d \vphi.
\end{equation}
\begin{definition}[Class of operators]
	Given $s, \sigma \geq 0$, we consider the space 
		\begin{equation}
	{\cal M}^{s}_{\sigma_1, \sigma_2}:={\cal H}^{s}\left(\T^{\frak n};\ \bhsSunoSdue{\sigma_1}{\sigma_2} \right),
	\end{equation}
	endowed with the norm
	\begin{equation}\label{classe operatori lineari riducibilita}
	\begin{aligned}
	& \| { R}\|\indSsigsig{s}{\sigma_1}{\sigma_2} := \Big( \sum_{l \in \Z^{\frak n}} \langle l  \rangle^{2 s} \big(\|\widehat{R}(l ) \|^{HS}_{\sigma_1, \sigma_2}\big)^2 \Big)^{\frac12}\,.
	\end{aligned}
	\end{equation}
\end{definition}

\begin{definition}[Higher regularity norm]
	Let $\Omega_0 \subseteq \Omega$ and $R\in \Lipsp\left(\Omega_0; {\cal M}^{s}_{\s_1, \s_2}\right)$. Given $\beta >0,$ if $R(\tildeom)$ is such that 
	$$ {R}(\tildeom) \in \Lipsp\left(\Omega_0; {\cal M}^{s+\beta}_{\s_1, \s_2}\right), \quad {\Gradbeta {R}(\tildeom) \in \Lipsp\left(\Omega_0; {\cal M}^{s}_{\s_1, \s_2}\right),}$$ we define
	\begin{equation} \label{def norme beta}
	\normabeta{ R}{\sigma_1}{\sigma_2}:= \| R\|^{\Lip} \indSsigsig{s+\beta}{\s_1}{\s_2} + \| \Gradbeta R\|^{\Lip} \indSsigsig{s}{\s_1}{\s_2}\,.
	\end{equation}
\end{definition}

\begin{definition}[Cutoffs]
	Given an operator ${ R} : L^2(\T^d) \to L^2(\T^d)$, for any $N \in \N$, we define the projector $\pi_N { R}$ as 
	\begin{equation}\label{proiettore pi N}
	(\pi_N {R})_j^{j'} := \begin{cases}
	{R}_j^{j'} \quad \textit{ if } |j-j'| <N\\
	0 \qquad \textit{ if } |j-j'| \geq N
	\end{cases}
	\end{equation}
	and we set $\pi_N^\bot {R} := { R} - \pi_N { R}$.
	For ${ R} : \T^{\frak n} \to {\cal B}(L^2(\T^d))$, $\vphi \mapsto {R}(\vphi)$, we define $\Pi_N { R}$ as 
	\begin{equation}\label{proiettore Pi N}
	\Pi_N { R}(\vphi) := \sum_{|l | \leq N} \pi_N \widehat{ R}(l )\, e^{\ii l \cdot \vphi}\,. 
	\end{equation}
	We then set $\Pi_N^\bot {R} := { R} - \Pi_N { R}$.
\end{definition}

In the following lemma we point out a key estimate for the remainder $\Pi_N^\bot { R}$ of an operator~${ R}:$
\begin{lemma} \label{lemma pienne bot} 
	Let ${ R}(\tildeom) \in {\cal M}^{s}_{\s_1, \s_2}$, $\tildeom \in \Omega_0 \subseteq \Omega$. Then for any $N>0$,
	\begin{equation} \label{stima pin ovvia}
	\|\Pi_N {R} \|^{\Lip}\indSsigsig{s}{\s_1}{\s_2}, \  \|\Pi_N^\bot {R} \| ^{\Lip}\indSsigsig{s}{\s_1}{\s_2} \leq \| {R}\|^{\Lip}\indSsigsig{s}{\s_1}{\s_2}.
	\end{equation}
	\noindent
	Moreover, let $\beta > 0$ and assume that ${R}(\tildeom) \in
        {\cal M}^{s+\beta}_{\s_1, \s_2}$, $\Gradbeta {R}(\tildeom) \in
        {\cal M}^{s}_{\s_1, \s_2}$, $\tildeom \in
        {\tilde{\Om}}$. Then, for any $N \in \N$, one has $\Pi_N^\bot
        {R}(\tildeom) \in {\cal M}^{s}_{\s_1, \s_2}$ and
	\begin{equation} \label{stima pin meno ovvia}
	\normaSsigsig{ \Pi_N^\bot R}{\s_1}{\s_2} \leq N^{- \beta} \normabeta{R}{\s_1}{\s_2}
	\end{equation}
\end{lemma}

\begin{proof}
Estimate \eqref{stima pin ovvia} is a direct consequence of the
definitions \eqref{classe operatori lineari
  riducibilita}-\eqref{proiettore Pi N}. We prove estimate \eqref{stima pin meno ovvia}.  By \eqref{proiettore pi N}, \eqref{proiettore Pi N}, one has 
	\begin{equation}\label{splitting Pi N bot cal R}
	\begin{aligned}
	& \Pi_N^\bot { R}(\vphi)  = { R}_{1, N}(\vphi) + {R}_{2, N}(\vphi)\,, \\
	& { R}_{1, N}(\vphi) :=  \sum_{|l | \leq N}  \pi_N^\bot \widehat{R}(l )  e^{\ii l  \cdot \vphi}\,, \quad {R}_{2, N}(\vphi) :=  \sum_{|l | > N} \widehat{ R}(l )  e^{\ii l \cdot \vphi}\,. 
	\end{aligned}
	\end{equation}
	We estimate separately the two terms in the above formula. 
	
	\noindent
	{\sc Estimate of ${R}_{1, N}$.} For any $\ell \in \Z^{\frak n}$, one has 
	\begin{align*}
	\left(\normahs{\pi^{\bot}_{N} \widehat{R}(l )}{\s_1}{\s_2}\right)^2 &=
	\sum_{\begin{subarray}{c}
		k, k' \in \Z^d \\
		|k-k'| > N
		\end{subarray}} |\widehat{R}(l )_{k}^{k'}|^2 \langle k \rangle ^{2 \sigma_2} \langle k' \rangle ^{-2 \sigma_1}\\
	& \leq N^{-2\beta} \sum_{k, k' \in \Z^d} \langle k -k' \rangle ^{2\beta}|\widehat{R}( l )_{k}^{k'}|^2 \langle k \rangle ^{2 \sigma_2} \langle k' \rangle ^{-2 \sigma_1}\\
	&=  N^{-2\beta} \left(\normahs{\Gradbeta \widehat{R}(l ) }{\s_1}{\s_2}\right)^2.
	\end{align*}
	Therefore, recalling \eqref{classe operatori lineari riducibilita}, one gets the estimate 
	\begin{equation}\label{stima cal R N 1}
	\| {R}_{1, N} \|\indSsigsig{s}{\s_1}{\s_2} \leq N^{- \beta} \| \langle \nabla\rangle^\beta  {R} \|\indSsigsig{s}{\s_1}{\s_2}\,. 
	\end{equation}
	
	\noindent
	{\sc Estimate of ${R}_{2, N}$.} The operator ${ R}_{2, N}$ can be estimated as
	\begin{align}
	\left(\|{R}_{2, N} \|\indSsigsig{s}{\s_1}{\s_2}\right)^2 & = \sum_{|l | > N} \langle l  \rangle^{2 s}  \left(\normahs{\ \widehat{R}(l )}{\s_1}{ \s_2}\right)^2 \nonumber\\
	& \leq N^{- 2 \beta} \sum_{l \in \Z^{\frak n}} \langle l  \rangle^{2 (s + \beta)}  \left(\normahs{\ \widehat{R}(l)}{\s_1}{ \s_2}\right)^2\nonumber\\
	&= N^{- 2 \beta} \left(\| {R}\|\indSsigsig{s+\beta}{\s_1}{\s_2}\right)^2, \nonumber
	\end{align}
	implying that 
	\begin{equation}\label{stima cal R N 2}
	\|{R}_{2, N} \|\indSsigsig{s}{\s_1}{\s_2} \leq N^{- \beta} \| { R}\|\indSsigsig{s+\beta}{\s_1}{\s_2}\,. 
	\end{equation}
	The claimed inequality then follows by \eqref{def norme beta}, \eqref{splitting Pi N bot cal R}, \eqref{stima cal R N 1} and \eqref{stima cal R N 2}. 
\end{proof}

\subsection{Diagonalization}\label{sotto sezione riducibilita}
Fix $M>0$ and consider the matrix representation of the regularized operator $H^{(M)}$ of Theorem \ref{regula}, namely
\begin{equation}\label{definizione A0 + P0}
A_0  + P_0(\vphi), \quad A_0 := D_0 +  Z
\end{equation}
where $D_0,\ Z$ and $ P_0$ are the matrix representations of $\muzero({\tildeom}) \cdot \nabla$, $\e Z^{(M)}$ and ${\cal W}^{(M)}$ respectively. \\
Since $\muzero \cdot \nabla$ and $Z^{(M)}$ depend only on $\nabla$ and not on the $x$ variable, their associated operators $D_0$ and $Z$ remain diagonal if we pass to Fourier variables, so that we deal with the sum of a diagonal operator $A_0 = D_0 + Z$ and a perturbative term $ P_0(\vphi)$ whose dependence on the angle $\vphi$ we want to eliminate. More precisely 
\begin{equation}\label{parte diagonale inizio riducibilita}
A_0 = {\rm diag}_{j \in \Z^d} \lambda^{(0)}_j, \quad \lambda^{(0)}_j := \ii \muzero \cdot j + z(j)
\end{equation}
where we recall that $z \in {\cal L}ip(\Omega_{0, \gamma};  OPS^{1 - \frak e})$.
Before to state the reducibility theorem, we fix some constants. Given $\tau > 0$ we define
\begin{equation}\label{vincoli costanti alpha beta}
\alpha :=  12 \tau + 7\,, \quad \beta := \alpha + 1\,, \quad m := 2\tau + 2 
\end{equation} 
Moreover, we fix the scale on which we perform the reducibility scheme as
\begin{equation}\label{definizione Nk}
		N_k = N_0^{\left(\frac{3}{2}\right)^k} \quad \forall k \in \N, \quad N_{-1} : =1
\end{equation}
where for convenience we link $N_0$ and $\gamma$ as 
\begin{equation}\label{N0 gamma - 1}
N_0 = \gamma^{- 1}
\end{equation}
where $\gamma$ is the constant appearing in the definition \eqref{cantor regolarizzazione} of the set $\Omega_{0, \gamma}$  (see also \eqref{Omgj} in the theorem below). 
We also fix the number $M $ of regularization steps in Theorem \ref{regula} as 
\begin{equation}\label{definizione M reg rid}
M := 2 m + 2 \beta + [d/2] + 1\,. 
\end{equation}

\begin{remark}
By Theorem \ref{regula} one has that $P_0 = \e {\cal W}^{(M)} \in {\cal C}^\infty(\T^{\frak n}; OPS^{-M})$. Since by \eqref{definizione M reg rid}, $M > 2 m + 2 \beta + \frac{d}{2}$, by applying Lemma \ref{lemma norma dm}, one has that 
\begin{equation}\label{stime P0 prima di riducibilita}
\normaSsigsig{P_0 }{\s -m}{\s +m}\,,\, \normabeta{ P_0}{\s -m}{\s +m} \lesssim_{s, \sigma} \e\,, \quad \forall s \geq 0\,, \quad \forall \sigma>0\,. 
\end{equation}
\end{remark}

\begin{theorem}{{\bf (KAM reducibility)}} \label{thm:abstract linear reducibility}  
Consider the system \eqref{regolare}. Let $\gamma \in (0, 1)$, $\tau > 0$. Then for any $s > [\frak n/2] + 1, \sigma \geq 0$
there exist constants $C_0 = C_0(s, \sigma, \tau) > 0$ large enough and $\delta = \delta(s, \sigma,  \tau)  \in (0, 1)$ small enough such that, if
\begin{equation}\label{piccolezza1}
N_0^{C_0}\e  \leq \delta 
\end{equation}
then, for all $ k \geq 0 $:

\begin{itemize}
\item[${\bf(S1)}_k$] 
There exists a vector field  
\be\label{def:Lj}
H_k (\vphi) :=A_k + P_k (\vphi)\,, \qquad \vphi \in \T^\nu\,,
\ee
\begin{equation}\label{cal D nu}
A_k = {\rm diag}_{j \in \Z^d} \lambda_j^{(k)}, \quad \lambda_j^{(k)}({\tildeom}) = \lambda_j^{(0)}({\tildeom}) + \rho_j^{(k)} ({\tildeom})
\end{equation}
defined for all $ \widetilde \omega \in {\cal O}_{k, \gamma}$, where we set ${\cal O}_{0, \gamma} := \Omega_{0, \gamma}$ (see \eqref{cantor regolarizzazione}) and 
 for $k \geq 1$, 
\begin{align}
{\cal O}_{k, \gamma} & :=  \Big\{\widetilde \omega = (\omega, \nu) \in {\cal O}_{k - 1, \gamma}  : |\ii \omega \cdot l + \lambda_j^{(k - 1)}(\widetilde \omega) - \lambda_{j'}^{(k -1)}(\widetilde \omega)| \geq \frac{\gamma}{\langle l \rangle^\tau \langle j \rangle^\tau \langle j'\rangle^\tau} \nonumber\\
& \forall (l , j, j') \neq (0, j, j), \quad |l|, |j - j'| \leq N_{k - 1}\Big\}\,. \label{Omgj}
\end{align}
For $k \geq 0$, the Lipschitz functions ${\cal O}_{k, \gamma} \to \C$, ${\tildeom} \mapsto \rho_j^{(k)}({\tildeom})$, $j \in \Z^d$ satisfy 
\begin{equation}  \label{rjnu bounded}
 \sup_{j \in \Z^d} \langle j \rangle^{2 m} |\rho_j^{(k)}|^\Lip \lesssim_{s, \sigma} \e\,.  
\end{equation}
There exist a  constant $C_* = C_*(s, \sigma, \beta, \tau, m) > 0$ such that
\begin{align}
\normaSsigsig{P_k }{\s -m}{\s +m} \leq C_* N_{k-1}^{-\alpha} \e , \quad \normabeta{ P_k}{\s -m}{\s +m} \leq  C_* N_{k-1} \e \,.     \label{stima ricorsiva Pn beta}
\end{align}
Moreover, for $k \geq 1 $, 
\be	\label{Lnu+1}
H_{k} (\vphi)= (\Phi_{k - 1})_{\omega*} H_{k - 1}(\vphi)  \, , \quad \Phi_{k -1} := {\rm Id} + X_{k - 1} \, 
\ee
where the map $ X_{k-1} $  satisfies the estimates
\begin{gather} \label{stima Xn}
	\normaSsigsig{ X_{k-1} }{\s \pm m}{\s \pm m} \lesssim_s N_k^{4 \tau + 2} N_{k - 1}^{- \alpha} \e  \,. 
	\end{gather}
	Moreover, if $P_0(\vphi)$ is real and reversible, for any $k \geq 1$, $P_k(\vphi)$ is real and reversible and
\begin{equation} \label{i lambda_n real}
 \lambda^{(k)}_{j} \in \ii \R \quad \forall j \in \Z^d.
\end{equation}

	\item[${\bf(S2)}_{k}$] 
For all $j \in \Z^d$,  there exists a Lipschitz extension to the set $\Omega_{0, \gamma}$ defined in \eqref{cantor regolarizzazione}, that we denote by  
$ \widetilde \lambda_j^{(k)}: \Omega_{0, \gamma} \to \C $ of  $ \lambda_j^{(k)} : {\cal O}_{k, \gamma} \to \C$ 
satisfying,  for $k \geq 1$, 
\be\label{lambdaestesi}  
| \widetilde \lambda_j^{(k)} - \widetilde \lambda_j^{(k - 1)} |^\Lip  \lesssim  \langle j \rangle^{- 2m} \normaSsigsig{P_{k-1} }{\s -m}{\s +m}  \lesssim_{s, \sigma}   \langle j \rangle^{- 2m} N_{k-2}^{-\alpha} \e  \,.
\ee
 \end{itemize}
\end{theorem}
We remark that ${\bf (S2)}_k$ will be used to construct the final
eigenvalues $\lambda_j^{(\infty)} $. The procedure will be to show
that as $k\to\infty$, the sequence $\lambda_j^{(k)} $ admits a limit
on $\Omega_{0,\gamma}$ and then to use the final value
$\lambda_j^{(\infty)} $ in order to define the set in which
reducibility holds (c.f. eq. {\eqref{Cantor finale}}).

\subsection{Proof of Theorem \ref{thm:abstract linear reducibility}}
{\sc Proof  of ${\bf ({S}i)}_{0}$, $i=1, 2$.} 
Properties \eqref{def:Lj}-\eqref{stima ricorsiva Pn beta} hold by setting $\rho_j^{(0)} = 0$ for any $j \in \Z^d$, $ N_{-1} := 1 $ and recalling the estimate \eqref{stime P0 prima di riducibilita}. 

${\bf({S}2)}_0 $ holds, since the constant $\lambda_j^{(0)}$ is already defined for all ${\tildeom} \in \Omega_{0, \gamma}$ and in the real and reversible case it satisfies $ \lambda_j^{(0)} \in \ii \R$  in force of Proposition \ref{esempio1}. Thus we simply set $\rho_j^{(0)} = 0$ for any $j \in \Z^d$. 

\subsubsection{The reducibility step: proof of ${\bf ({S}i)}_{k + 1}$, $i=1, 2$.}

\noindent
{\sc Proof of ${\bf ({S}1)}_{k + 1}$.}

\noindent
We now describe the inductive step, showing how to define a transformation 
$ \Phi_k := {\rm Id} +   X_k  $ 
so that the transformed vector field $H_{k +1 }(\vphi) = (\Phi_k)_{\omega*} H_k(\vphi) $ has the desired properties. If we perform a change of coordinates of the form $u':= \Phi_k(\vphi)u,$ $\Phi_k(\vphi) = {\rm Id} + X_k(\vphi)$ one has that $H_{k + 1}(\vphi) = (\Phi_k)_{\omega*} H_k(\vphi) $ takes the form 
\begin{align}
H_{k + 1}(\vphi) &= A_k+ \Phi_k(\vphi)^{- 1}\big( \Pi_{N_k} P_k(\vphi) + [X_k(\vphi), A_k] - \omphi X_k(\vphi) \big)  \nonumber\\
&+ \Phi_k(\vphi)^{- 1}\big( \Pi_{N_k}^{\bot} P_k(\vphi) + P_k(\vphi) X_k(\vphi) \big) \nonumber
\end{align}
We look for a transformation $X_k(\vphi)$ solving the {\it homological equation}
\begin{equation} \label{hom eq operatori}
\Pi_{N_k} P_k(\vphi) + [X_k(\vphi), A_k] - \omphi X_k(\vphi) = \overline{P_k}
\end{equation}
where $\overline{P_k}$ is a diagonal operator. Then we set
\begin{equation}\label{roba nuova}
\begin{aligned}
& A^{k + 1} = A_k + \overline{P_k}, \quad P^{k + 1} = \Pi_{N_k}^\bot P_k + P_k X_k + (\Phi_k^{- 1} - {\rm Id}) \big( \overline P_k + \Pi_{N_k}^\bot P_k + P_k X_k \big)   \,, \\
& \overline P_k := {\rm diag}_{j \in \Z} (\widehat P_k)_j^j(0)\,. 
\end{aligned}
\end{equation}
By formula \eqref{roba nuova} one obtains that 
$$
A_{k + 1} := {\rm diag}_{j \in \Z^d} \lambda_j^{(k + 1)}
$$
where for any $j \in \Z^d$
\begin{equation}\label{definizione autovalore k + 1}
\begin{aligned}
& \lambda_j^{(k + 1)} := \lambda_j^{(k)} + \widehat P_k(0)_{j}^j =\ii  \muzero \cdot j + \e z( j) + \rho_{j }^{(k + 1)} \\
& \rho_j^{(k + 1)} := \rho_j^{(k)} + \widehat P_k(0)_{j}^j\,.
\end{aligned}
\end{equation}
In the real and reversible case, since $P_k$ is real and reversible, by Lemma \ref{rmk real rev anti sa fourier} one has $\widehat P_k(0)_j^j \in \ii \R$, and since $\lambda_j^{(k)}, \rho_j^{(k)} \in \ii \R$ then one has that $\lambda_j^{(k + 1)}, \rho_j^{(k + 1)} \in \ii \R$. 

\medskip

\noindent
By the definition \eqref{definizione autovalore k + 1}, applying Lemma \ref{lemma decay operatore smooth} and using the estimate \eqref{stima ricorsiva Pn beta}, one gets that for any $j \in \Z^d$ for any $i \in \{ 0, 1, \ldots, k\}$
\begin{align}
|\lambda_j^{(i + 1)} - \lambda_j^{(i)}|^{\rm Lip} & = |\rho_j^{(i + 1)} - \rho_j^{(i)}|^{\rm Lip}  = | (\widehat P_i)_{jj}(0)|^{\rm Lip}  \nonumber\\
& \lesssim  \langle j \rangle^{- 2 m} \| P_i\|_{{\cal M}^s_{\sigma - m, \sigma + m}}^{\rm Lip}  \lesssim_{s, \sigma}  \langle j \rangle^{- 2 m} N_{i-1}^{-\alpha} \e\,. \label{stima rho k + 1 - rho k}
\end{align}
We now verify the estimate \eqref{rjnu bounded} at the step $k + 1$. By using a telescoping argument, recalling that $\rho_j^{(0)} = 0$ for any $j \in \Z^d$, one gets that 
\begin{align}
|\rho_j^{(k + 1)}|^{\rm Lip} & \leq \sum_{i = 0}^k |\rho_j^{(i + 1)} - \rho_j^{(i)}|^{\rm Lip} \stackrel{\eqref{stima rho k + 1 - rho k}}{\lesssim_{s, \sigma}} \langle j \rangle^{- 2 m} \e \sum_{i = 0}^\infty N_{i-1}^{-\alpha}  \lesssim \langle j \rangle^{- 2 m} \e
\end{align}
since the series $\sum_{i = 0}^\infty N_{i-1}^{-\alpha}$ is convergent (see \eqref{definizione Nk}). Hence \eqref{rjnu bounded} is verified at the step $k + 1$. 

\medskip

\noindent
In the next lemma we will show how to solve the homological equation
\eqref{hom eq operatori}. This is the main lemma of the section. 
\begin{lemma} \label{lemma_sol_eq_homol}
	Let $m >2 \tau + 1$. Then for any ${\tildeom} \in {\cal O}_{k + 1, \gamma}$ (recall \eqref{Omgj}), the homological equation
	\begin{equation} \label{hom_eq}
	[A_k, X_k] + \omphi X_k = \Pi_{N_k} P_k - \overline{P}_k,
	\end{equation}
	with
	\begin{equation} \label{media}
	\overline{P}_k = \diag_{j \in \Z^d} \widehat{P_k}(0)_{j}^{j} 
	\end{equation}
	has a solution $X_k$ defined on ${\cal O}_{k, \gamma}$ and satisfying the estimates 
	\begin{eqnarray} \label{stima x p}
	\| X_k \|^{\Lip}\indSsigsig{s}{\s \pm m}{\s \pm m} \lesssim  N_k^{4 \tau + 2} \| P_k\|^{\Lip}\indSsigsig{s}{\s -m}{\s +m},\\ \label{stima grad x p}
	\| \Gradbeta X_k \|^{\Lip}\indSsigsig{s}{\s \pm m}{\s \pm m} \lesssim N_k^{4 \tau + 2} \| \Gradbeta P_k\|^{\Lip}\indSsigsig{s}{\s - m}{\s + m}.
	\end{eqnarray}
	Furthermore, if $P_k$ is real and reversible then $X_k$ is real and reversibility preserving. 
\end{lemma}

\begin{proof}
	
	To simplify notations, here we drop the index $ k $, namely we write $A$, $P$, $X$, $\lambda_j$, $\rho_j$ instead of $A_k$, $P_k$, $X_k$, $\lambda_j^{(k)}$, $\rho_j^{(k)}$. 
	Taking the $(j, j')$ matrix element and the $l-$th Fourier coefficient of \eqref{hom_eq} we get:
	\begin{align*}
	&  (\ii \omega \cdot l + \lambda_j - \lambda_{j '}) \ \widehat X(l)_j^{j'} =  \widehat P(l)_j^{j'} \quad \textrm{if } 0 <|j-j'|< N,\ 0<|l|<N\\
	&\widehat X(l)_j^{j'} = 0 \quad \textrm{otherwise} 
	\end{align*}
	Since ${\tildeom} \in {\cal O}_{k+1, \gamma}$ one has 
	\begin{align} \label{x p diofantea}
	| \widehat X(l)_j^{j'}| \leq \frac{|\widehat P(l)_j^{j'}||j|^{\tau} |j'|^{\tau} |l|^{\tau}}{\gamma},
	\end{align}
	hence
	\begin{equation} \label{stima x p eltwise 1} 
	\begin{aligned}
	| \widehat X(l)_j^{j'}| &\lesssim  \gamma^{- 1} | \widehat P(l)_j^{j'}| |l|^{\tau}\langle j'\rangle ^{\tau}\Big(\langle j'\rangle ^{\tau} + |j-j'|^{\tau}\Big)\\
	&\leq \gamma^{- 1} | \widehat P(l)_j^{j'}| N^{\tau}  \langle j'\rangle ^{\tau}\Big(\langle j'\rangle^{\tau} + N^{\tau}\Big)\\
	&\lesssim \gamma^{-1} | \widehat P(l)_j^{j'}| N^{2\tau} \langle j'\rangle^{2\tau},
	\end{aligned}
	\end{equation}
	Similarly, one gets
	\begin{align} \label{stima x p eltwise 2}
	| \widehat X(l)_j^{j'}|&\lesssim \gamma^{-1} | \widehat P(l)_j^{j'}| N^{2\tau} \langle j \rangle ^{2\tau}.
	\end{align}
	Thus, recalling that  $\tau<m,$ (see \eqref{vincoli costanti alpha beta}) the norm $\| X\|\indSsigsig{s}{\s +m}{\s +m}$ is estimated by:
	\begin{equation} \label{stima x p conti +}
	\begin{aligned}
	\left(\| X\|\indSsigsig{s}{\s +m}{\s +m}\right)^2&=
	\sum_{l \in \Z^{\frak n}} \langle l \rangle ^{2 s} \sum_{j, j' \in \Z^d} \langle j \rangle ^{2 (\s +m)} |\widehat X(l)_j^{j'}(l)|^2 \langle j' \rangle^{-2(\s + m)}\\
	&\lesssim \gamma^{-2}  N^{4 \tau }\sum_{l \in \Z^{\frak n}} \langle l \rangle ^{2 s} \sum_{j,j' \in \Z^d} \langle j \rangle ^{2 (\s +m)} |\widehat P(l)_j^{j'}|^2 \langle j'\rangle^{4\tau} \langle j' \rangle^{-2(\s + m)}\\
	&\leq \gamma^{-2}  N^{4 \tau }\sum_{l \in \Z^{\frak n}} \langle l \rangle ^{2 s}  \sum_{j, j' \in \Z^d} \langle j \rangle ^{2 (\s +m)} |\widehat P(l)_j^{j'}|^2 \langle j' \rangle^{-2(\s -m)} \\
	&= \gamma^{-2}  N^{4 \tau } \left(\| P\|\indSsigsig{s}{\s-m}{\s +m}\right)^2.
	\end{aligned}
	\end{equation}
	Similarly, one obtains 
	\begin{equation} \label{stima x p conti -}
	\begin{aligned}
	\left(\| X\|\indSsigsig{s}{\s-m}{\s -m}\right)^2 \lesssim \gamma^{-2}  N^{4 \tau } \left(\| P\|\indSsigsig{s}{\s-m}{\s +m}\right)^2.
	\end{aligned}
	\end{equation}
	 To estimate the norm of the operator $\Gradbeta X, $ we argue as in \eqref{stima x p eltwise 1}, \eqref{stima x p eltwise 2}  to get
	\begin{equation}
	\begin{gathered}
	\langle j-j' \rangle^{\beta} | \widehat X(l)_j^{j'}|\lesssim N^{2\tau} \langle j\rangle ^{2\tau}\langle j-j' \rangle^{\beta}  |  \widehat P(l)_j^{j'}|,\\
	\langle j-j' \rangle^{\beta} | \widehat X(l)_j^{j'}|\lesssim N^{2\tau} \langle j'\rangle ^{2\tau} \langle j-j' \rangle^{\beta}  |  \widehat P(l)_j^{j'}|;
	\end{gathered}
	\end{equation}
	hence we repeat the same argument of \eqref{stima x p conti +}, \eqref{stima x p conti -} to get \eqref{stima grad x p}.
Concerning Lipschitz estimates, recall that the eigenvalues $\lambda_j$, $j \in \Z^d$ have the expansion
$$
\lambda_j({\tildeom}) = \lambda^{(0)}_j({\tildeom}) + \rho_j({\tildeom}) = \ii \muzero ({\tildeom}) \cdot j + z({\tildeom},  j) + \rho_j({\tildeom})\,.
$$
By \eqref{tordo4b}, \eqref{z lip 1} and the induction hypotheses \eqref{rjnu bounded} one has that for any $\widetilde \omega_1, \widetilde \omega_2 \in \Omega_{\gamma}$ and any $j , j' \in \Z^d$, one has 
\begin{equation}\label{stima lambda j lip lemma eq omo}
|(\lambda_j - \lambda_{j'})(\widetilde \omega_1) - (\lambda_j - \lambda_{j'})(\widetilde \omega_2)| \lesssim \e \gamma^{- 1} \langle j - j' \rangle | \widetilde \omega_1 -  \widetilde \omega_2|\,. 
\end{equation} 
 Hence, one uses $|l|, |j - j'| \leq N$, \eqref{x p diofantea}, \eqref{stima lambda j lip lemma eq omo} and the inequality
 \begin{align*}
	|l|^{2\tau+1}|j|^{2\tau}|j'|^{2\tau} &\lesssim_{\tau}
	N^{2\tau+1} |j|^{2\tau}\left(|j|^{2\tau}+ N^{2\tau}\right) \lesssim N^{4 \tau + 1} \langle j \rangle^{4 \tau}
	\end{align*}
to deduce the Lipschitz estimates as usual.
	By Lemma \ref{rmk real rev anti sa fourier} of the Appendix, if $A = {\rm diag}_{j \in \Z^d} \lambda_j$ and $P$ are real and reversible one easily get that $X$ is real and reversible too. 
\end{proof}
The estimate \eqref{stima Xn} follows from \eqref{stima x p} and \eqref{stima ricorsiva Pn beta}. Moreover, using that by \eqref{vincoli costanti alpha beta}, $\alpha > 6 \tau + 3$ and by using the smallness condition \eqref{piccolezza1}, one gets that 
\begin{equation}\label{stima Xk leq 1}
\| X_k\|_{{\cal M}^s_{\sigma \pm m, \sigma \pm m}}^{\rm Lip} \leq \delta(s) 
\end{equation}
for some $\delta(s) \in (0,1)$ small enough. Therefore, one can apply Lemma \ref{Neumann series} implying that 
\begin{equation}\label{stima Phi k inv - Id}
\begin{aligned}
& \| \Phi_k^{- 1} - {\rm Id} \|_{{\cal M}_{\sigma \pm m, \sigma \pm m}^s}^{\rm Lip} \lesssim_{s, \sigma} \| X_k\|_{{\cal M}_{\sigma \pm m, \sigma \pm m}^s}^{\rm Lip} \stackrel{\eqref{stima x p}}{\lesssim_s} N_k^{4 \tau + 2} \| P_k\|^{\Lip}\indSsigsig{s}{\s -m}{\s +m} \\
& \| \langle \nabla \rangle^\beta(\Phi_k^{- 1} - {\rm Id}) \|_{{\cal M}_{\sigma \pm m, \sigma \pm m}^s}^{\rm Lip} \lesssim_{s, \beta} \| \langle \nabla \rangle^\beta X_k\|_{{\cal M}_{\sigma \pm m, \sigma \pm m}^s}^{\rm Lip} \stackrel{\eqref{stima grad x p}}{\lesssim_{s, \beta}} N_k^{4 \tau + 2} \| \langle \nabla \rangle^\beta P_k\|^{\Lip}\indSsigsig{s}{\s -m}{\s +m}
\end{aligned}
\end{equation}
In the next lemma, we obtain key estimates for the remainder term $P_{k + 1}$ defined in \eqref{roba nuova}.
\begin{lemma} \label{rmk stime P+}
	There exists a constant $C = C(s, \sigma, \tau) > 0$ such that the operator $P_{k + 1}(\vphi)$ defined in \eqref{roba nuova} fulfills
	\begin{equation}\label{cornetto 0}
		\begin{aligned}
		& \normaSsigsig{P_{k + 1} }{\s -m}{\s +m} \leq C\Big( N_k^{4 \tau + 2} \big( \normaSsigsig{ P_k}{\s -m}{\s +m} \big)^2  
		+ N_k^{-\beta}  \normabeta{P_k}{\s -m}{\s +m} \Big), \\
		& \normabeta{P_{k + 1} }{\s -m}{\s +m} \leq C  \normabeta{ P_k}{\s -m}{\s +m}\,. 
		\end{aligned}
	\end{equation}
	Furthermore, if $P_k(\vphi)$ is real and reversible then $P_{k + 1}(\vphi)$ is real and reversible too.
\end{lemma}
\begin{proof}
	By recalling the definition of $P_{k +1}$ given in \eqref{roba
          nuova}, using the inductive estimates \eqref{stima x p},
        \eqref{stima grad x p}, and the estimate \eqref{stima Phi k
          inv - Id}, by applying Lemma \ref{lemma pienne bot} and
        Lemma \ref{stima composizione} in the appendix, which gives an
        estimate of the product of operators, we get
	\begin{align} 
		\normaSsigsig{P_{k + 1} }{\s -m}{\s +m} &\lesssim_{s, \sigma} N_k^{4 \tau + 2} \big(\normaSsigsig{ P_k}{\s -m}{\s +m} \big)^2  \nonumber\\
		& \quad + N_k^{-\beta} \left(\| P_k\|^{\Lip} \indSsigsig{s+\beta}{\s -m}{\s +m} + \normaSsigsig{ \Gradbeta P_k}{\s -m}{\s +m}\right),\label{stima P k + 1 s}\\ \label{stima s+b}
		\| P_{k + 1}\|^{\Lip} \indSsigsig{s+\beta}{\s -m}{\s +m} &\lesssim_{s, \sigma} N_k^{4 \tau + 2}  \normaSsigsig{ P_k}{\s -m}{\s +m} \| P_k\|^{\Lip} \indSsigsig{s+\beta}{\s -m}{\s +m} + \| P_k\|^{\Lip} \indSsigsig{s+\beta}{\s -m}{\s +m}, \\ 
		\normaSsigsig{ \Gradbeta P_{k + 1}}{\s -m}{\s +m} &\lesssim_{s, \sigma}  \normaSsigsig{ \Gradbeta P_k}{\s -m}{\s +m} \nonumber\\
		& \quad + N_k^{4 \tau + 2}  \normaSsigsig{ P_k}{\s -m}{\s +m} \normaSsigsig{ \Gradbeta P_k}{\s -m}{\s +m} . \label{stima gradb} 
	\end{align}
	Recalling that
	$
	\normabeta{\cdot}{\s -m}{\s +m }= \| \cdot\|^{\Lip} \indSsigsig{s+\beta}{\s -m}{\s +m} + \normaSsigsig{ \Gradbeta \cdot}{\s -m}{\s +m}
	$
	and summing up the contribution of \eqref{stima s+b}, \eqref{stima gradb}, we get
	\begin{equation}\label{cornetto}
	\begin{aligned}
		\normaSsigsig{P_{k + 1} }{\s -m}{\s +m} &\lesssim N_k^{4 \tau + 2} \big(\normaSsigsig{ P_k}{\s -m}{\s +m} \big)^2  
		+ N_k^{-\beta} \normabeta{P_k}{\s -m}{\s +m},\\ 
		\normabeta{P_{k + 1}}{\s -m}{\s +m} &\lesssim N_k^{4 \tau + 2}  \normaSsigsig{ P_k}{\s -m}{\s +m} \normabeta{P_k}{\s -m}{\s +m} + \normabeta{P_k}{\s -m}{\s +m}.
	\end{aligned}
	\end{equation}
	Furthermore, by using the smallness condition \eqref{piccolezza1}, recalling the definition \eqref{definizione Nk}, using that $\alpha > 6 \tau + 3$, taking $N_0$ large enough and $\e$ small enough one gets that
	$$
	N_k^{4 \tau + 2}  \normaSsigsig{ P_k}{\s -m}{\s +m} \stackrel{}{\lesssim} N_k^{4 \tau + 2} N_{k - 1}^{- \alpha}\e  \leq 1\
	$$
and then \eqref{cornetto} implies the claimed estimate \eqref{cornetto 0}. 
	
	\noindent
	Finally, if $P_k$ is real and reversible, then by Lemma \ref{lemma_sol_eq_homol}, the operator $X_k$ (and hence $\Phi_k = {\rm Id} + X_k$ and $\Phi_k^{- 1}$) is real and reversibility preserving.  By the definition \eqref{roba nuova}, one concludes that $P_{k + 1}$ is real and reversible. 
\end{proof}
By Lemma \ref{rmk stime P+} one has 
\begin{align}
\normabeta{P_{k + 1} }{\s -m}{\s +m} & \leq C  \normabeta{ P_k}{\s -m}{\s +m} \stackrel{\eqref{stima ricorsiva Pn beta}}{\leq} C C_* \e N_{k - 1} \leq C_* \e N_k \nonumber
\end{align}
provided $CN_{k - 1} \leq N_k$ for any $k \geq 0$. This latter condition is verified by taking $N_0 > 0$ large enough. Furthermore 
\begin{align}
\normaSsigsig{P_{k + 1} }{\s -m}{\s +m} & \leq C N_k^{4 \tau + 2} \big( \normaSsigsig{ P_k}{\s -m}{\s +m} \big)^2  
		+C N_k^{-\beta}  \normabeta{P_k}{\s -m}{\s +m} \nonumber\\
		& \stackrel{\eqref{stima ricorsiva Pn beta}}{\leq} C N_k^{4 \tau + 2} C_*^2 \e^2 N_{k - 1}^{- 2 \alpha} + C N_k^{- \beta} C_* N_{k - 1} \e \leq C_* \e N_k^{- \alpha} \nonumber
\end{align}
provided 
$$
2 C N_k^{\alpha + 4 \tau + 2} N_{k - 1}^{- 2 \alpha} \e \leq 1\,, \quad \,, 2 C N_k^{\alpha - \beta} N_{k - 1} \leq 1 \quad \forall k \geq 0\,. 
$$
The above conditions are verified by \eqref{vincoli costanti alpha beta}, the smallness condition \eqref{piccolezza1}, recalling the definition \eqref{definizione Nk} and taking $\e$ small enough and $N_0$ large enough. Hence the estimate \eqref{stima ricorsiva Pn beta} is proved at the step $k + 1$. The proof of ${\bf ({S}1)}_{k + 1}$ is then concluded. 

\medskip

\noindent	
{\sc Proof  of ${\bf ({S}2)}_{k + 1}$.} By the estimate \eqref{stima rho k + 1 - rho k}, on the set ${\cal O}_{k, \gamma}$, $\delta_j^{(k)} := \rho_j^{(k + 1)} - \rho_j^{(k)}$ satisfies $|\delta_j^{(k)}|^{\rm Lip} \lesssim \langle j \rangle^{- 2 m} \| P_k\|_{{\cal M}^s_{\sigma - m, \sigma + m}}^{\rm Lip}  \lesssim_{s, \sigma}  \langle j \rangle^{- 2 m} N_{k-1}^{-\alpha} \e$ for any $j \in \Z^d$. By the Kirszbraun Theorem (see Lemma M.5 in \cite{KapPoe}), we extend the function $\delta_j^{(k)} : {\cal O}_{k, \gamma} \to \C$ to a function $\widetilde \delta_j^{(k)} : \Omega_{0, \gamma} \to \C$ which still satisfies the estimate $|\widetilde \delta_j^{(k)}|^{\rm Lip} \lesssim \langle j \rangle^{- 2 m} \| P_k\|_{{\cal M}^s_{\sigma - m, \sigma + m}}^{\rm Lip}  \lesssim_{s, \sigma}  \langle j \rangle^{- 2 m} N_{k-1}^{-\alpha} \e$. Therefore, ${\bf ({S}2)}_{k + 1}$ follows by defining $\widetilde \rho_j^{(k + 1)} := \widetilde \rho_j^{(k)} + \widetilde \delta_j^{(k)}$ and $\widetilde \lambda_j^{(k + 1)} = \lambda_j^{(0)} + \widetilde \rho_j^{(k + 1)}$ (note that $\lambda_j^{(0)}$ is already defined on $\Omega_{0, \gamma}$). Note that in the real and reversible case, one has that $\rho_j^{(k)}, \lambda_j^{(k)} : {\cal O}_{\gamma, k- 1} \to \ii \R$, $\widetilde \rho_j^{(k)}, \widetilde \lambda_j^{(k)} : \Omega_{ 0, \gamma} \to \ii \R$, $\delta_j^{(k)} : {\cal O}_{k, \gamma} \to \ii \R$ and hence $\widetilde \lambda_j^{(k + 1)}\,,\, \widetilde \rho_j^{(k +1)} : \Omega_{0, \gamma} \to \ii \R$. 
	
\subsection{Passing to the limit and completing the diagonalization procedure}	

By Theorem \ref{thm:abstract linear reducibility}-${\bf ({S}2)}_{k}$,
using a telescoping argument, for any $j \in \Z^d$, the sequence
$(\widetilde \rho_j^{(k)})_{k \geq 0}$ is a Cauchy sequence w.r. to
the norm $| \cdot |^{\rm Lip}$ in $\Omega_{0, \gamma}$, and hence it converges to
$\rho_j^{(\infty)}$. The following estimates hold:
\begin{equation}\label{stime autovalori finali}
|\widetilde \rho_j^{(k)} - \rho_j^{(\infty)}|^{\rm Lip} \lesssim_{s, \sigma} \langle j \rangle^{- 2m} N_{k-1}^{-\alpha} \e\,,  \qquad \, |\rho_j^{(\infty)}|^{\rm Lip} \lesssim_{s, \sigma} \langle j \rangle^{- 2m}  \e\,. 
\end{equation}
Note that in the real and reversible case, $\rho_j^{(\infty)} : \Omega_{0, \gamma} \to \ii \R$ for any $j \in \Z^d$.

\noindent
We then define the {\it final eigenvalues} $\lambda_j^{(\infty)} : \Omega_{0, \gamma} \to \C$ as 
\begin{equation}\label{definizione autovalori finali}
\lambda_j^{(\infty)} := \lambda_j^{(0)} + \rho_j^{(\infty)} \stackrel{\eqref{parte diagonale inizio riducibilita}}{=} \ii \muzero \cdot j + z(j) + \rho_j^{(\infty)}\,, \quad j \in \Z^d\,. 
\end{equation}
We then define 
\begin{equation}\label{Cantor finale}
\begin{aligned}
{\cal O}_{\infty, \gamma} & :=  \Big\{\widetilde \omega = (\omega, \nu) \in \Omega_{0, \gamma} \ :  \ |\ii \omega \cdot l + \lambda_j^{(\infty)}(\widetilde \omega) - \lambda_{j'}^{(\infty)}(\widetilde \omega)| \geq \frac{2\gamma}{\langle l \rangle^\tau \langle j \rangle^\tau \langle j'\rangle^\tau} \\
& \forall (l , j, j') \neq (0, j, j) \quad \Big\}\,. 
\end{aligned}
\end{equation}
The following lemma holds.
\begin{lemma}\label{inclusione cantor finale}
One has ${\cal O}_{\infty, \gamma} \subseteq \cap_{k \geq 0} {\cal O}_{k, \gamma}$. 
\end{lemma}
\begin{proof}
We prove by induction that for any $k \geq 0$ one has ${\cal O}_{\infty, \gamma} \subseteq {\cal O}_{k, \gamma}$. For $k = 0$, it follows by definition that ${\cal O}_{\infty, \gamma} \subseteq {\cal O}_{0, \gamma}$ since ${\cal O}_{0, \gamma} = \Omega_{0, \gamma}$. Then assume that ${\cal O}_{\infty, \gamma} \subseteq {\cal O}_{k, \gamma}$ for some $k \geq 0$ and let us show that ${\cal O}_{\infty, \gamma} \subseteq {\cal O}_{k + 1, \gamma}$. Let $\widetilde \omega= (\omega, \nu) \in {\cal O}_{\infty, \gamma}$. Since by the induction hypothesis $\widetilde \omega \in {\cal O}_{k, \gamma}$ one has that by Theorem \ref{thm:abstract linear reducibility}-${\bf (S1)}_k$, $\lambda_j^{(k)}(\widetilde \omega)$ is well defined and by Theorem \ref{thm:abstract linear reducibility}-${\bf (S2)}_k$ one has that $\widetilde \lambda_j^{(k)}(\widetilde \omega) = \lambda_j^{(k)}(\widetilde \omega)$ and $\widetilde \rho_j^{(k)}(\widetilde \omega) = \rho_j^{(k)}(\widetilde \omega)$ (recall that $\lambda_j^{(k)} = \lambda_j^{(0)} + \rho_j^{(k)}$ and $\widetilde \lambda_j^{(k)} = \lambda_j^{(0)} + \widetilde \rho_j^{(k)}$).  
We then have that for any $(l, j, j') \neq (0, j, j)$, $|l|, |j - j'| \leq N_k$, 
\begin{align}
|\ii \omega \cdot l + \lambda_j^{(k)}(\widetilde \omega) - \lambda_{j'}^{(k)}(\widetilde \omega)| & \geq |\ii \omega \cdot l + \lambda_j^{(\infty)}(\widetilde \omega) - \lambda_{j'}^{(\infty)}(\widetilde \omega)| - |\widetilde \rho_j^{(k)}(\widetilde \omega) - \rho_j^{(\infty)}(\widetilde \omega) |  \nonumber\\
& \quad -   |\widetilde \rho_{j'}^{(k)}(\widetilde \omega) - \rho_{j'}^{(\infty)}(\widetilde \omega) | \nonumber\\
& \stackrel{\eqref{stime autovalori finali}}{\geq}  \frac{2 \gamma}{\langle l \rangle^\tau \langle j \rangle^\tau \langle j' \rangle^\tau} -  \frac{C \e}{ N_{k - 1}^\alpha{\rm min}\{ \langle j \rangle , \langle j' \rangle \}^{2m}}  \nonumber\\
& \geq \frac{\gamma}{\langle l \rangle^\tau \langle j \rangle^\tau \langle j' \rangle^\tau} \nonumber
\end{align}
provided 
\begin{equation}\label{condizione inclusione nel lemma}
\frac{C \e \langle l \rangle^\tau \langle j \rangle^\tau \langle j' \rangle^\tau}{\gamma N_{k - 1}^\alpha {\rm min}\{ \langle j \rangle , \langle j' \rangle \}^{2m} } \leq 1\,. 
\end{equation}
Using that $|l|, |j - j'| \leq N_k$, $m > \tau$ and since 
$$
\langle j \rangle \langle j' \rangle \leq \big( \langle j - j' \rangle + {\rm min}\{ \langle j \rangle, \langle j' \rangle \} \big)^2 \lesssim \langle j - j' \rangle^2 + {\rm min}\{ \langle j \rangle, \langle j' \rangle \}^2  \lesssim N_k^2 + {\rm min}\{ \langle j \rangle, \langle j'\rangle \}^2
$$ 
one gets that 
\begin{align}
\frac{\langle l \rangle^\tau \langle j \rangle^\tau \langle j' \rangle^\tau}{{\rm min}\{ \langle j \rangle , \langle j' \rangle \}^{2m}} & \lesssim   N_k^{3 \tau}\,.
\end{align}
Therefore 
$$
\frac{C \e \langle l \rangle^\tau \langle j \rangle^\tau \langle j' \rangle^\tau}{\gamma N_{k - 1}^\alpha {\rm min}\{ \langle j \rangle , \langle j' \rangle \}^{2m} } \leq C' \e \gamma^{- 1} N_k^{3 \tau} N_{k - 1}^{- \alpha} \leq 1
$$
since $\alpha > \frac92 \tau$ (see \eqref{vincoli costanti alpha beta}) and by taking $\e$ small enough (see the smallness condition \eqref{piccolezza1} and recall that $\gamma^{- 1} = N_0$). Condition \eqref{condizione inclusione nel lemma} is then verified and hence $\widetilde \omega \in {\cal O}_{k + 1, \gamma}$. This concludes the proof of the lemma. 
\end{proof}
For any $k \geq 0$, $\widetilde \omega \in {\cal O}_{\infty, \gamma}$ we define the map 
\begin{equation}\label{definizione mappe iterate riducibilita}
{\cal V}_k(\vphi, \widetilde \omega) \equiv {\cal V}_k(\vphi) := \Phi_0(\vphi) \circ \Phi_1(\vphi) \circ \ldots \circ \Phi_k(\vphi)\,. 
\end{equation}
Note that by Lemma \ref{inclusione cantor finale} and Theorem \ref{thm:abstract linear reducibility} all the maps $\Phi_k(\vphi)$ are well defined for $\widetilde \omega \in {\cal O}_{\infty, \gamma}$.

\noindent
The following lemma holds
\begin{lemma}\label{lemma trasformazione finale KAM}
The sequence $({\cal V}_k)_{k \geq 0}$ converges to an invertible operator ${\cal V}_\infty$ in ${\cal L}ip\Big({\cal O}_{\infty, \gamma}; {\cal H}^s\big( \T^{\frak n}; {\cal B}(
 {\cal H}^{\sigma \pm m}, {\cal H}^{\sigma \pm m}\big) \Big)$ and the operator ${\cal V}_\infty^{\pm 1} - {\rm Id}$ satisfies the estimate 
$$
\| {\cal V}_\infty^{\pm 1} - {\rm Id}\|^{\rm Lip}_{H^s\big( \T^{\frak n}, {\cal B}(
 {\cal H}^{\sigma \pm m}, {\cal H}^{\sigma \pm m}\big)}   \lesssim_{s, \sigma} N_0^{4 \tau + 2} \e\,. 
$$
Moreover in the real and reversible case, ${\cal V}_\infty^{\pm 1}$ is real and reversibility preserving. 
\end{lemma}
\begin{proof}
The proof is based on standard arguments and therefore it is omitted (see for instance the proof of Corollary 4.1 in \cite{Mon17a}). The presence of $N_0^{4 \tau + 2}$ in front of $\e$ in the claimed inequality is due to the fact that \eqref{stima Xn} for $k = 0$ gives $\| \Phi_0 - {\rm Id}\|^{\rm Lip}_{{\cal M}^s_{\sigma \pm m, \sigma \pm m}} \lesssim_{s, \sigma} N_0^{4 \tau + 2} \e$. 
\end{proof}
\begin{lemma}\label{final conjugacy}
For any $\widetilde \omega \in {\cal O}_{\infty, \gamma}$, one has that $({\cal V}_\infty)_{\omega*} (A_0 +P_0) = H_\infty$ (recall \eqref{definizione A0 + P0}) where the operator $H_\infty$ is given by $H_\infty = {\rm diag}_{j \in \Z^d} \lambda_j^{(\infty)}$. Furthermore in the real and reversible case, the eigenvalues $\lambda_j^{(\infty)}$ are purely imaginary. 
\end{lemma}
\begin{proof}
By \eqref{Lnu+1} and recalling the definition \eqref{definizione mappe iterate riducibilita}, one gets that for any $k \geq 1$
$$
  ({\cal V}_{k - 1})_{\omega*}(A_0 + P_0(\vphi)) = H_k(\vphi) = A_k + P_k(\vphi)\,. 
$$
The claimed statement then follows by passing to the limit in the above identity, recalling the definition of $A_k$ given in \eqref{cal D nu}, the definition \eqref{definizione autovalori finali}, the estimates \eqref{stima ricorsiva Pn beta}, \eqref{stime autovalori finali} and Lemma \ref{lemma trasformazione finale KAM}. 
\end{proof}
\subsection{Measure Estimates}\label{sezione stime misura}
In this section we show that the set ${\cal O}_{\infty, \gamma} $ defined in \eqref{Cantor finale} has {\it large} Lebesgue measure. We prove the following 
\begin{proposition}\label{proposizione stima di misura}
One has $|\Omega \setminus {\cal O}_{\infty, \gamma}| \lesssim\gamma$. 
\end{proposition}
Since $\Omega \setminus {\cal O}_{\infty, \gamma} = (\Omega \setminus \Omega_{0, \gamma}) \cup (\Omega_{0, \gamma} \setminus {\cal O}_{\infty, \gamma})$ and by Remark \ref{misura Omega 0 gamma} one has that $|\Omega \setminus \Omega_{0, \gamma}| \lesssim \gamma$, it is enough to estimate the measure of the set $\Omega_{0, \gamma} \setminus {\cal O}_{\infty, \gamma}$. By the definition \eqref{Cantor finale}, one has that 
\begin{equation}\label{definizione risonanti}
\begin{aligned}
& \Omega_{0, \gamma} \setminus {\cal O}_{\infty, \gamma} = \bigcup_{\begin{subarray}{c}
(l, j, j') \in \Z^{\frak n} \times \Z^d \times \Z^d \\
(l, j - j') \neq (0, 0)
\end{subarray}} {\cal R}_{l j j'}(\gamma) \\
& {\cal R}_{l j j'}(\gamma) := \Big\{{\tildeom} =(\omega, \nu) \in \Omega_{0, \gamma} : |\ii \omega \cdot l + \lambda_j^{(\infty)}( \omega, \nu) - \lambda_{j'}^{(\infty)}( \omega, \nu)| < \frac{2 \gamma}{\langle l \rangle^\tau \langle j \rangle^\tau \langle j' \rangle^\tau} \Big\}
\end{aligned}
\end{equation}
\begin{lemma}\label{stima singolo risonante}
One has $|{\cal R}_{l j j'}(\gamma)| \lesssim  \gamma \langle l \rangle^{- \tau} \langle j \rangle^{- \tau} \langle j' \rangle^{- \tau}$
\end{lemma}
\begin{proof}
By \eqref{definizione autovalori finali}, one has that for any $j \in \Z^d$
$$
\lambda_j^{(\infty)}(\omega, \nu) = \ii \muzero(\omega, \nu) \cdot j + z(j, \omega, \nu ) + \rho_j^{(\infty)}(\omega, \nu)
$$
where by the estimates \eqref{tordo4b}, \eqref{z lip 1}, one has $|\muzero - \nu|^{\Lipg} \lesssim \e$, $\sup_{j \in \Z^d}|\partial_\xi z(j)|^{\rm lip} \lesssim \e$. Then the map 
$$
\Psi : \Omega_{0, \gamma} \to \Psi(\Omega_{0, \gamma}), \quad (\omega, \nu) \mapsto (\omega, \muzero(\omega, \nu))
$$
is a Lipschitz homeomorphism with inverse given by 
$
\Psi^{- 1} : \Psi(\Omega_{0, \gamma}) \to \Omega_{0, \gamma}\,, \quad (\omega, \zeta) \mapsto \Psi^{- 1}(\omega, \zeta) 
$
and satisfying 
\begin{equation}\label{stima lip Psi - 1}
|\Psi^{- 1} - {\rm Id}|^{\rm sup} \lesssim \e\,, \quad |\Psi^{- 1} - {\rm Id}|^{\rm lip} \lesssim \e \gamma^{- 1}\,. 
\end{equation}
 Defining 
$$
a_j^{(\infty)}(\omega, \zeta) := \lambda_j^{(\infty)}(\Psi^{- 1}(\omega, \zeta)), \quad j \in \Z^d
$$
and 
$$
\widetilde{\cal R}_{l j j'}(\gamma) := \Big\{ (\omega, \zeta) \in \Psi (\Omega_{0, \gamma}) : |\ii \omega \cdot l + a_j^{(\infty)}(\omega, \zeta) - a_{j'}^{(\infty)}(\omega, \zeta)| < \frac{2 \gamma}{\langle l \rangle^\tau \langle j \rangle^\tau \langle j' \rangle^\tau} \Big\}
$$
one has that 
\begin{equation}\label{misura cal R cal R tilde}
|{\cal R}_{l j j'}(\gamma)| \simeq |\widetilde{\cal R}_{l j j'}(\gamma)|,
\end{equation}
 then we estimate the measure of the set $\widetilde{\cal R}_{l j j'}(\gamma)$. The functions $a_j^{(\infty)}$ admit the expansion 
$$
a_j^{(\infty)}(\omega, \zeta) = \ii \zeta \cdot j + z_\Psi(j, \omega, \zeta) + r_j^{(\infty)}(\omega, \zeta)
$$
where 
$$
z_\Psi(j, \omega, \zeta) := z(j, \Psi^{- 1}(\omega, \zeta)), \quad r_j^{(\infty)}(\omega, \zeta) := \rho_j^{(\infty)}(\Psi^{- 1}(\omega, \zeta))\,. 
$$
By the estimate \eqref{stima lip Psi - 1} and using the estimates \eqref{z lip 1}, \eqref{stime autovalori finali} on $z$ and $\rho_j^{(\infty)}$, for $\e \gamma^{- 1}$ small enough, one can easily deduce that 
\begin{equation}\label{stima coefficienti autovalori con inverso}
\sup_{j \in \Z^d}|\partial_\xi z_\Psi(j, \cdot)|^{\rm Lip} \lesssim \e, \quad  \sup_{j \in \Z^d} \langle j \rangle^{2 m} |r_j^{(\infty)}|^{\rm Lip} \lesssim \e\,. 
\end{equation}
Since $(l, j - j') \neq (0, 0)$, we write 
$$
(\omega, \zeta) = (\omega(s), \zeta(s)))= \frac{(l, j - j')}{|(l, j - j')|} s + w, \quad w \in \R^{\frak n + d}, \quad w \cdot (l, j - j') = 0
$$
and we consider 
$$
\begin{aligned}
f_{l j j'}(s) & := \ii \omega(s) \cdot l + a_j^{(\infty)}(\omega(s), \zeta(s)) - a_{j'}^{(\infty)}(\omega(s), \zeta(s)) \\
& = \ii |(l, j - j')| s + z_\Psi(j, \omega(s), \zeta(s)) - z_\Psi(j', \omega(s), \zeta(s)) + r_j^{(\infty)}(\omega(s), \zeta(s)) - r_{j'}^{(\infty)}(\omega(s), \zeta(s))\,.
\end{aligned}
$$
Using the estimates \eqref{stima coefficienti autovalori con inverso} one obtains that 
\begin{align}
|f_{l j j'}(s_1) - f_{l j j'}(s_2)| & \geq \Big( |(l, j - j')| - C\e |j  - j'| - C \e \Big) |s_1 - s_2| \nonumber\\
& \stackrel{|j - j'| \leq |(l, j - j')|}{\geq}  \Big( (1 - C \e )|(l, j - j')|  - C \e \Big) |s_1 - s_2| \geq \frac12 |s_1 - s_2|
\end{align}
by taking $\e$ small enough. This implies that 
$$
\Big| \big\{ s : |f_{l j j'}(s)| < \frac{2 \gamma}{\langle l \rangle^\tau \langle j \rangle^\tau \langle j' \rangle^\tau} \big\} \Big| \lesssim \frac{\gamma}{\langle l \rangle^\tau \langle j \rangle^\tau \langle j' \rangle^\tau}\,. 
$$
By a Fubini argument one gets that $|\widetilde{\cal R}_{l j j'}(\gamma)| \lesssim \gamma \langle l \rangle^{- \tau} \langle j \rangle^{- \tau} \langle j' \rangle^{- \tau}$. The claimed statement then follows by recalling \eqref{misura cal R cal R tilde}. \end{proof}
\noindent {\sc Proof of Proposition \ref{proposizione stima di
    misura}}. By \eqref{definizione risonanti} and Lemma \ref{stima
  singolo risonante} one gets that $$|\Omega_{0, \gamma} \setminus
          {\cal O}_{\infty, \gamma}| \lesssim \gamma \sum_{l \in
            \Z^{\frak n}, j, j' \in \Z^d} \langle l \rangle^{- \tau}
          \langle j \rangle^{- \tau} \langle j' \rangle^{- \tau}
          \lesssim \gamma$$ since $\tau > {\rm max}\{ \frak n\,,\,
          d\}$. The claimed statement then follows by recalling that
          $|\Omega \setminus \Omega_{0, \gamma}| \lesssim \gamma$ and
          that $\Omega \setminus {\cal O}_{\infty, \gamma} = (\Omega
          \setminus \Omega_{0, \gamma}) \cup (\Omega_{0, \gamma}
          \setminus {\cal O}_{\infty, \gamma})$.

\subsection{Proof of Theorem \ref{main teo}}
We consider the composition
$$
{\cal U}(\vphi) = {\cal V}(\vphi) \circ {\cal V}_\infty(\vphi), \quad {\cal V}(\vphi) : ={\cal A}(\vphi) \circ e^{- \e \calg_1(\vphi, \tildeom)} \circ \cdots \circ \ e^{- \e \calg_M(\vphi, \tildeom)},
$$ where ${\cal A}(\vphi)$ is defined in Section \ref{sezione
  riduzione ordine alto}, the maps $e^{- \e {\cal G}_K}$ are
constructed in Section \ref{sezione regolarizzazione} (see Theorem
\ref{regula}) and ${\cal V}_\infty$ is given in Lemma \ref{lemma
  trasformazione finale KAM}. By Section \ref{sezione riduzione ordine
  alto}, Theorem \ref{regula} and Lemma \ref{final
  conjugacy}, for any $\widetilde \omega \in {\cal O}_{\infty,
  \gamma}$, the map ${\cal U}(\vphi)$ conjugates the equation
\eqref{main equation} to the equation $\partial_t u = H_\infty u$
where $H_\infty$ is the diagonal operator with eigenvalues
$(\lambda_j^{(\infty)})_{j \in \Z^d}$. Let $0 < \frak a <
\frac{1}{C_0}$ and $N_0 := \frac{1}{\e^{\frak a}}$ so that the
smallness condition \eqref{piccolezza1}, i.e. $N_0^{C_0} \e \leq
\delta$ becomes
$$
N_0^{C_0} \e = \e^{1 - C_0 \frak a} \leq \delta\,,
$$
which is satisfied for $\e$ small enough. Since $\gamma = N_0^{- 1} = \e^{\frak a}$, setting $\Omega_\e := {\cal O}_{\infty, \gamma}$, Proposition \ref{proposizione stima di misura} implies that $\lim_{\e \to 0} |\Omega \setminus \Omega_\e| = 0$. 
The proof is therefore concluded.  
\subsection{Proof of Corollary \ref{corollario caso rev real}}

By Theorem \ref{main teo}, for any $\widetilde \omega = (\omega , \nu)
\in \Omega_\e$ under the change of coordinates $u = {\cal U}(\omega t)
v$, the Cauchy problem 
\begin{equation}\label{Cauchy problem trasporto}
\begin{cases}
\partial_t u = \Big( \nu + \e V(\omega t, x) \Big) \cdot \nabla u + \e {\cal W}(\omega t)[u] \\
u(0, x) = u_0(x),
\end{cases} \qquad u_0 \in {\cal H}^\sigma(\T^d) 
\end{equation} 
is transformed into 
\begin{equation}\label{Cauchy problem trasporto ridotto}
\begin{cases}
\partial_t v = H_\infty v \\
v(0) = v_0,
\end{cases}\qquad v_0 := {\cal U}(0)^{- 1} u_0\,. 
\end{equation}
Using that for any $\widetilde \omega = (\omega, \nu) \in \Omega_\e$, ${\cal U}(\vphi)$ is bounded and invertible on ${\cal H}^\sigma$ one gets that 
\begin{equation}\label{gandalf il grande}
\| \psi\|_{{\cal H}^\sigma} \lesssim_\sigma \| {\cal U}(\vphi)^{\pm 1} \psi \|_{{\cal H}^\sigma} \lesssim_\sigma \| \psi\|_{{\cal H}^\sigma}, \quad \forall \psi \in {\cal H}^\sigma(\T^d)
\end{equation}
uniformly w.r. to $\vphi \in \T^{\frak n}$. 

\noindent
{\sc Case $(1)$.} If all the eigenvalues $\lambda_j^{(\infty)}$, $j \in \Z^d$ of the operator $H_\infty$ are purely imaginary, the solution of the Cauchy problem \eqref{Cauchy problem trasporto ridotto} satisfies $\| v(t, \cdot) \|_{{\cal H}^\sigma} = \| v_0\|_{{\cal H}^\sigma}$ for any $t \in \R$. By the estimate \eqref{gandalf il grande} and recalling that $u = {\cal U}(\omega t) v$ one obtains the desired bound on the solution $u(t, x)$ of \eqref{Cauchy problem trasporto}. 

\noindent
{\sc Case (2)} Let $j \in \Z^d$ so that ${\rm Re}(\lambda_j^{(\infty)}) \neq 0$. Then for any $\alpha \in \C$, the solution $v$ of the Cauchy problem \eqref{Cauchy problem trasporto ridotto} with initial datum $v_0(x) = \alpha e^{\ii j \cdot x}$ is given by 
$$
v(t, x) = \alpha e^{\lambda_j^{(\infty)} t} e^{\ii j \cdot x}\,.
$$
Hence, setting $u_0 := {\cal U}(0)[\alpha e^{\ii j \cdot x}] = \alpha {\cal U}(0)[e^{\ii j \cdot x}]$, one has that the solution of the Cauchy problem \eqref{Cauchy problem trasporto} with such an initial datum $u_0$ is given by 
$$
u(t, x) = {\cal U}(\omega t)[\alpha e^{\lambda_j^{(\infty)} t} e^{\ii j \cdot x}] = \alpha e^{\lambda_j^{(\infty)} t} {\cal U}(\omega t)[e^{\ii j \cdot x}]\,. 
$$
Recalling \eqref{gandalf il grande} one gets that 
$$
\| u(t, \cdot) \|_{{\cal H}^\sigma} \simeq_\sigma C_j e^{{\rm Re}(\lambda_j^{(\infty)}) t}\,.
$$
This gives the growth for $t>0$ if Re$\lambda_j^{(\infty)}>0 $ or for
$t<0$ if Re$\lambda_j^{(\infty)}>0 $. If there exists
$\lambda_j^{(\infty)}$ with  Re$\lambda_j^{(\infty)}>0 $ and
$\lambda_{j'}^{(\infty)}$ with  Re$\lambda_{j'}^{(\infty)}<0 $ then
the solution with initial datum $\alpha e^{\ii j \cdot x}+\beta
e^{\ii j' \cdot x}$ grows both as $t>0$ and as $t<0$.

\newpage



\appendix

\section{Appendix}\label{sezione appendice}

To regularize \eqref{main equation}, we make use of operators that are the flow at time $\tau \in [-1,\ 1]$ of the PDE
$$
\partial_{\tau} u =  { G}(\vphi) u
$$
for a given pseudo differential operator $G(\vphi) \in OPS^{\eta},\ \eta \leq 1.$
An operator of this sort is denoted by~$e^{ \tau G}.$
Thus, we state some of its main properties. The proof is a variant of
Proposition A.2 of \cite{MaRo}.

\begin{lemma}\label{lemma flusso PDE}
	Let $\eta< 1$ and ${ G}  \in \cinfty\left(\T^{\frak n}; OPS^{\eta}\right)$ be such that ${ G}(\vphi) + { G}(\vphi)^* \in OPS^0$ and let $e^{ \tau { G}}$ be the flow  of the autonomous PDE $\ \partial_{\tau} u = {G}(\vphi) u, \quad \tau \in [-1, 1].$\\
	$(i)$ Then $e^{ \tau { G}}(\vphi) \in \bSunoSdue{\sigma}{\sigma} \ \forall \s >0.$
	\\
	$(ii)$ $\forall \s >0,\ \forall\ \alpha \in \N^{{\frak n}},\ $ $\partial_{\vphi}^{\alpha} e^{\tau { G}}(\vphi) \in \bSunoSdue{\sigma}{\sigma - \eta |\alpha|}.$\\
	$(iii)$ If $G \in \Lipsp \left(\Om;\ \cinfty\left(\T^{\frak n}; OPS^{\eta}\right) \right),\ $ $ \partial_{\vphi}^{\alpha} e^{\tau { G}}(\vphi, \om) \in \Lipsp\left( \Om;\ \bSunoSdue{\sigma}{\sigma - \eta |\alpha| - \eta}\right)\quad \forall \s >0,\ \forall\ \alpha \in \N^{{\frak n}}.$\\
	Furthermore, if $G$ is reversibility preserving (or real), $e^{\tau G}$ is reversibility preserving (resp. real) too.
\end{lemma}
\begin{proof}
	Item $(i)$ is a well known result. It is proved trough a Galerkin type approximation on the subspace
	$E_N$ of the compact supported sequences $\{\hat{u}_k \}_{k \in \Z^d}$ such that $\hat{u}_k =0 \quad \textrm{ if } |k|>N.$ See \cite{Taylor}, Section $0.8$, for details. 
	\\
	Items $(ii)$ and $(iii)$ follow as in Lemma A.3 in \cite{BertiMontalto}. 
	\\{\sc Reversibility preserving property:}
	We remark that since 
	$$
	S \circ \partial_\tau = \partial_\tau \circ S,
	$$
	one both has
	\begin{align*}
	\partial_\tau [ S  \circ e^{\tau G(\vphi)}] u = S \circ \partial_\tau \circ e^{\tau G(\vphi)}  u = S \circ  { G}(\vphi) e^{\tau G(\vphi)} u = { G}_{\lambda} (-\vphi) \circ S\ u 
	\end{align*}
	and
	\begin{align*}
	\partial_\tau [e^{\tau G(-\vphi)} \circ S] u = { G} (-\vphi) \circ S\ u.
	\end{align*}
	Since $S  \circ e^{\tau G(\vphi)}$ and $e^{\tau G(-\vphi)} \circ S$ solve the same initial value problem for all the functions $u(x),$ they must coincide. Thus we can deduce the reversibility preserving property for $e^{\tau G(\vphi)}.$\\
	{\sc Reality:} the proof of the reality can be done arguing similarly, using that since $G = \overline G$, then $e^{\tau G(\vphi)}$ and $\overline{e^{\tau G(\vphi)}}$ solve the same initial value problem. 
\end{proof}
Let $a :  [- 1, 1] \times \T^{\frak n} \times \T^d  \to \R^d$, $(\tau, \vphi, x) \mapsto a(\tau, \vphi, x)$ be a ${\cal C}^\infty$ function and let us consider the transport equation
\begin{equation}\label{trasporto appendice}
\partial_\tau u = a(\tau, \vphi, x) \cdot \nabla u\,. 
\end{equation}
We denote by $\Phi(\tau_0, \tau, \vphi )$ the flow of the above PDE. For convenience, we set $\Phi(\tau, \vphi) \equiv \Phi(0, \tau, \vphi)$.  The following lemma holds: 
\begin{lemma}\label{flusso trasporto}
$(i)$ For any $\tau_0, \tau \in [0, 1]$ the flow $\Phi( \tau_0, \tau, \vphi)$ of the equation \eqref{trasporto appendice} is a bounded linear operator on the Sobolev space ${\cal H}^s(\T^d)$ for any $s \geq 0$. Moreover the map $\vphi \mapsto \Phi( \tau_0, \tau, \vphi)$ is differentiable and for any $\alpha \in \N^{\frak n}$, the map $\T^{\frak n} \to {\cal B}({\cal H}^{s + |\alpha|}, {\cal H}^s)$, $\vphi \mapsto \partial_\vphi^\alpha \Phi( \tau_0, \tau, \vphi)$ is bounded. 

\noindent
$(ii)$ Assume that $a = a(\widetilde \omega , \tau, \vphi, x)$, $(\widetilde \omega , \tau, \vphi, x) \in \Omega \times [0, 1] \times \T^{\frak n} \times \T^d$ is in ${\cal L}ip\Big(\Omega\,,\, {\cal C}^\infty([0, 1] \times \T^{\frak n} \times \T^d, \R^d) \Big)$. Then for any $\alpha \in \N^{\frak n}$ the map 
$$
\T^{\frak n} \to {\cal L}ip\Big(\Omega, {\cal B}({\cal H}^{s + |\alpha| + 1}, {\cal H}^s) \Big), \quad \vphi \mapsto \partial_\vphi^\alpha \Phi(\tau_0, \tau, \widetilde \omega, \vphi)
$$
is bounded. 
\end{lemma}
\begin{remark} \label{commutat pseudodiff}
	Let $A(\vphi) \in \Lipsp \left( \Om;\ \cinfty(\T^{{\frak n}}; OPS^{m})\right)$ and ${ G} \in \Lipsp \left( \Om;\ \cinfty(\T^{{\frak n}}; OPS^{\eta})\right),$ with $\eta <1.$
	If $\forall\ j \in \N$ we define
	\begin{equation}\label{definizione Ad}
	Ad^{0}_{{ G}} A = A, \quad Ad^{j+1}_{{G}} A = [{ G},  Ad^{j}_{{ G}} A ],
	\end{equation}
	then
	$$
	Ad^{j}_{G} A \in \Lipsp \left(\Om;\ \cinfty\left(\T^{{\frak n}}; OPS^{ m -j(1-\eta)}\right)\right) \quad \forall\ j \in \N.
	$$
\end{remark}
The following simpler version of the Egorov theorem holds. 
\begin{lemma} \label{lemma coniugo col flusso}
	Let $A(\vphi) \in \Lipsp \left( \Om;\ \cinfty(\T^{{\frak n}}; OPS^{m})\right)$ and ${G} \in \Lipsp\left( \Om;\ \cinfty(\T^{{\frak n}}; OPS^{\eta})\right),$ with $\eta <1 $ and ${ G}$ such that ${G(\vphi)+ G(\vphi)^* \in OPS^0.}$ Then
	$$
	e^{\tau { G}} A e^{-\tau {G}} \in \Lipsp \left( \Om;\ \cinfty(\T^{{\frak n}}; OPS^{m})\right).
	$$
\end{lemma}
\begin{proof}
	This version of the Egorov theorem is actually simpler than the one stated in Theorem A.0.9 in \cite{Taylor}. The reason is that the order of $G$ is strictly smaller than one and hence one has the asymptotic expansion 
	$$
	e^{\tau { G}} A e^{-\tau { G}} \sim \sum_{j=0}^{\infty} A_j 
	$$
	with $A_j \in OPS^{ m -j(1-\eta)} $ (see remark \ref{commutat pseudodiff}). 
\end{proof}

\begin{remark} \label{lemma assumption ii}
	Note that by Theorem A.0.9 in \cite{Taylor} one has that if $A \in \Lipsp \left( \Om;\ \cinfty(\T^{\frak n}; OPS^{m})\right),$ then $e^{ i \tau \cdot K} A e^{- i \tau \cdot K},\ \partial_{\tau}^{\alpha} \left(e^{ i \tau \cdot K} A e^{- i \tau \cdot K}\right) \in \Lipsp \left( \Om;\ \cinfty(  \T^{\frak n}; OPS^{m})\right)$ $\ \forall \alpha \in \N^d.$
\end{remark}

\begin{lemma} \label{lemma tau rev}
	Given $S$ acting as $S : u(x) \mapsto u(-x)$, a linear operator $A(\vphi)$ satisfies the reversibility condition
	$$
	A(\vphi) \circ S = -S \circ A(-\vphi)
	$$
	if and only if ${\cal A}(\tau, \vphi) := e^{i \tau \cdot K} A(\vphi) e^{-i \tau \cdot K}$ satisfies the reversibility condition
	$$
	{\cal A}(\tau, \vphi) \circ S = -S \circ A(-\tau, -\vphi).
	$$
	Analogously, $A(\vphi)$ satisfies the reversibility preserving condition
	$$
	A(\vphi) \circ S = S \circ A(-\vphi)
	$$
	if and only if ${\cal A}(\tau, \vphi) := e^{i \tau \cdot K} A(\vphi) e^{-i \tau \cdot K}$ satisfies the reversibility preserving condition
	$$
	{\cal A}(\tau, \vphi) \circ S = S \circ A(-\tau, -\vphi).
	$$
	Furthermore, $A(\vphi)$ is real if and only if ${\cal A}(\tau, \vphi)$ is real.
\end{lemma}
\begin{proof}
	We only prove the statement concerning the reversibility. The statement on reality can be proved similarly.\\
	A direct calculation shows that $e^{i \tau \cdot K} \circ S = S \circ e^{-i \tau \cdot K}$, 
	hence, if $A(\vphi)$ is $\vphi$-reversible, one immediately gets 
	\begin{align*}
	{\cal A}(\tau, \vphi) S =  - S {\cal A}(-\tau, -\vphi).
	\end{align*}
	Vice versa, ${\cal A}(\tau, \vphi) \circ S = - S \circ {\cal A}(-\tau, -\vphi)$ implies (for $\tau = 0$)
	$$
	A(\vphi) \circ S = {\cal A}(0, \vphi) \circ S = - S \circ {\cal A}(0, -\vphi) = -S \circ A(-\vphi).
	$$
\end{proof}

\begin{lemma} \label{lemma commutatori adj}
	Let $\eta<1,\ G \in \cinfty(\T^{\frak n}; OPS^{\eta})$ with $G + G^{*} \in OPS^{0}$ and $A \in \cinfty(\T^{\frak n}; OPS^{1}).$ Then
	\begin{itemize}
		\item[(i)]
		$$
		Ad^{k}_{G}A + (Ad^{k}_{G}A)^* \in OPS^{-(k-1)(1-\eta)} \quad \forall\ k \geq 1;
		$$
		\item[(ii)] In particular,
		$$
		\left(e^{G} A e^{-G} - A\right) + \left(e^{G} A e^{-G} - A\right)^{*} \in OPS^{0}.
		$$
	\end{itemize}
\end{lemma}
\begin{proof}
	{\sc Proof of $(i)$.}We argue by induction: if $k=1,$ one has
	\begin{align*}
	[G, A] - [G^{*},\ A^*] &= [G, A] + [A^*, G^{*}]\\
	&= [G, A + A^{*} ] + [A^*,\ G + G^{*}] \in OPS^{0}.
	\end{align*}
	Assume that for some $k\geq 1$ 
	$$
	Ad^{k}_{G}A + (Ad^{k}_{G}A)^* \in OPS^{ -(k-1)(1-\eta)}\,.
	$$
	A direct calculation shows that 
	\begin{align*}
	Ad^{k + 1}_G A + (Ad_G^{k + 1} A)^* & = [G + G^*, Ad_G^k A] - [G^*, Ad_G^k A + (Ad_G^k A)^* ]\,.
	\end{align*}
	Since by Remark \ref{commutat pseudodiff} $Ad_G^k A, (Ad_G^k A)^* \in OPS^{1 - k(1 - \eta)}$ and using the induction hypothesis and that $G^* \in OPS^\eta$, $G + G^* \in OPS^0$, one obtains that $Ad^{k + 1}_G A + (Ad_G^{k + 1} A)^* \in OPS^{- k(1 - \eta)}$.
	
	\noindent
	{\sc Proof of $(ii)$.} 
	$\forall\ M>0$ one computes
	\begin{gather*}
	e^{-G} A e^{G} - A = \sum_{k=1}^{M} \frac{Ad^{k}_{G}A}{k!} + \int_{0}^{1} \frac{(1-s)^{M+1}}{(M+1)!} 	e^{-s G} Ad^{M+1}_{G}A e^{s G}.
	\end{gather*}
	By applying Remark \ref{commutat pseudodiff}, choosing $M$ large enough such that ${\eta -(1-M)(1-\eta) <0,}$ one gets that 
	\begin{align*}
	e^{-G} A e^{G} - A + \big( e^{-G} A e^{G} - A \big)^* & = \sum_{k=1}^{M} \frac{Ad^{k}_{G}A + (Ad^{k}_{G}A)^*}{k!} + OPS^0  \stackrel{item (i)}{\in} OPS^0\,. 
	\end{align*}
\end{proof}

\begin{lemma} \label{rmk real rev anti sa fourier}
	Let $P  \in {\cal M}^{s}_{\s_1, \s_2}$ and $\forall\ k, k' \in \Z^d,\ \forall\ l \in \Z^{\frak n}$ let $[\widehat{P}(l)]_{k}^{k'}$ be the $(k,k')-$th matrix element with respect to the basis $\{e^{i k \cdot x}\ |\ k \in \Z^d \}$ of the operator $\widehat{P}(l)$ defined as in \eqref{def fourier l}. The following conditions hold:
	\begin{itemize}
		\item[(a)]  $P(\vphi)$ is real if and only if
		$$
		[\widehat{P}(l)]_{k}^{k'}= \left([\widehat{P}(-l)]_{-k}^{-k'}\right)^{*};
		$$
		\item[(b)] $P(\vphi)$ is reversible if and only if
		$$
		[\widehat{P}(l)]_{k}^{k'}= -[\widehat{P}(-l)]_{-k}^{-k'};
		$$
		\item[(c)] $P(\vphi)$ is reversibility preserving if and only if
		$$
		[\widehat{P}(l)]_{k}^{k'}= [\widehat{P}(-l)]_{-k}^{-k'}\,. 
		$$
	\end{itemize}
\end{lemma}

\subsection{Tame estimates in ${\cal M}^{s}_{\s_1,\s_2}$}
\begin{lemma} \label{lemma prodotto 1} 
	$(i)$ Let $ \s_1, \s_2, \s_3 \in \R$ and let us assume that ${\cal R}$, ${\cal P}$ are linear operators such that\\
	 ${\cal P}\in \bhsSunoSdue{\s_1}{\s_2}$, ${\cal R} \in \bhsSunoSdue{\s_2}{\s_3},$  Then $ {\cal R \cal P} \in \bhsSunoSdue{\s_1}{\s_3} $  with
	$$
	\normahs{\cal R \cal P}{\s_1}{ \s_3} \leq \normahs{\cal R}{\s_2}{\s_3}  \normahs{\cal P}{\s_1}{ \s_2}.
	$$

\noindent
$(ii)$ Let $\sigma_1,  \sigma_2, \sigma_3 \in \R$, $\beta \geq 0$ and assume that 
$\langle \nabla \rangle^\beta{\cal P}, {\cal P}\in \bhsSunoSdue{\s_1}{\s_2}.$, $\langle \nabla \rangle^\beta{\cal R}, {\cal R} \in \bhsSunoSdue{\s_2}{\s_3},$  Then $ \langle \nabla \rangle^\beta {\cal R \cal P} \in \bhsSunoSdue{\s_1}{\s_3} $  with
	$$
	\| \langle \nabla \rangle^\beta {\cal R} {\cal P} \|_{\sigma_1, \sigma_3}^{HS}	\lesssim_\beta  \| \langle \nabla \rangle^\beta {\cal R}  \|_{\sigma_2, \sigma_3}^{HS} \| {\cal P} \|_{\sigma_1, \sigma_2}^{HS} + \|  {\cal R}  \|_{\sigma_2, \sigma_3}^{HS} \| \langle \nabla \rangle^\beta {\cal P} \|_{\sigma_1, \sigma_2}^{HS}
	$$
\end{lemma}
\begin{proof}
We prove the estimate $(ii)$. The estimate $(i)$ can be proved by similar arguments (and it is actually simpler). We have that
\begin{align*}
	\left(\normahs{ \Gradbeta {\cal R}{\cal P} }{\sigma_1}{ \sigma_3}\right)^2 & = \sum_{k, k', \in \Z^d} \langle k \rangle^{2 \sigma_3} \langle k' \rangle^{-2 \sigma_1} \big|\sum_{j \in \Z^d} \langle k - k' \rangle ^{\beta} {\cal R}_{k}^{j} {\cal P}_{j}^{k'}\big|^{2}\\
	&\lesssim_{\beta} \sum_{k, k', \in \Z^d} \langle k \rangle^{2\sigma_3} \langle k' \rangle^{-2 \sigma_1} \big[\sum_{j \in \Z^d} \left(\langle k - j \rangle ^{\beta} + \langle j-k' \rangle^{\beta}\right) |{\cal R}_{k}^{j} {\cal P}_{j}^{k'}|\big]^{2}\\
	&\lesssim_\beta \sum_{k, k', \in \Z^d} \langle k \rangle^{2 \sigma_3} \langle k' \rangle^{-2 \sigma_1} \big[\sum_{j \in \Z^d} |\left(\Gradbeta {\cal R}\right)_{k}^{j}|\ |{\cal P}_{j}^{k'}|\big]^{2}\\
	&+ \sum_{k, k', \in \Z^d} \langle k \rangle^{2 \sigma_3} \langle k' \rangle^{-2 \sigma_1} \big[\sum_{j \in \Z^d} |{\cal R}_{k}^{j}|\ |\left(\Gradbeta {\cal P}\right)_{j}^{k'}|\big]^{2}\\
	&\lesssim_\beta \sum_{k, j \in \Z^d} \langle k \rangle^{2\sigma_3} |\left(\Gradbeta {\cal R}\right)_{k}^{j}|^2 \langle j \rangle ^{-2 \sigma_2}  \sum_{j, k' \in \Z^d} \langle j \rangle ^{2 \sigma_2} |{\cal P}_{j}^{k'}|^{2} \langle k' \rangle^{-2 \sigma_1}\\
	&+ \sum_{k, j \in \Z^d} \langle k \rangle^{2 \sigma_3} |{\cal R}_{k}^{j}|^2 \langle j \rangle ^{-2 \sigma_2}  \sum_{j, k' \in \Z^d} \langle j \rangle ^{2 \s_2} |\left(\Gradbeta {\cal P}\right)_{j}^{k'}|^{2} \langle k' \rangle^{-2 \sigma_1}\\
	& \lesssim_\beta \left(\normahs{\Gradbeta \cal R}{ \s_2}{\s_3}\right)^2 \left(\normahs{\cal P}{\s_1}{\s_2}\right)^2 + \left(\normahs{ \cal R}{\s_2}{\s_3}\right)^2 \left(\normahs{\Gradbeta \cal P}{\s_1}{\s_2}\right)^2.
	\end{align*}

\end{proof}
\begin{lemma}\label{stima composizione} 
	\hfill \break
	$(i)$ Let $s \geq s_0$, $\sigma_1, \sigma_2, \sigma_3 \in \R$, $ {\cal P}(\lambda) \in {\cal M}^{s}_{\sigma_1, \sigma_2}$, $ {\cal R}(\lambda) \in {\cal M}^{s}_{\sigma_2,\sigma_3}$. Then $ {\cal R} {\cal P}(\lambda) \in {\cal M}^{s}_{\sigma_1, \sigma_3}$ and
	\begin{align*}
	\normaSsigsig{{\cal R} {\cal P}}{\sigma_1}{\sigma_3} \lesssim_{s}
	\| {\cal R} \|^{\Lip}\indSsigsig{s}{\sigma_2}{\sigma_3} \| {\cal P} \|^{ \Lip}\indSsigsig{s_0}{\sigma_1}{ \sigma_2} + \| {\cal R} \|^{\Lip}\indSsigsig{s_0}{\sigma_2}{ \sigma_3} \| {\cal P} \|^{ \Lip}\indSsigsig{s}{\sigma_1}{ \sigma_2}.
	\end{align*}
	\noindent
	$(ii)$ Let $\beta \geq 0$, $s \geq s_0$, $\sigma_1, \sigma_2, \sigma_3 \in \R$. Assume that $ \Gradbeta {\cal P}(\lambda) \in {\cal M}^{s}_{\sigma_1, \sigma_2}$,
	$ \Gradbeta {\cal R}(\lambda) \in  {\cal M}^{s}_{\sigma_2,\sigma_3}$. Then $ \Gradbeta {\cal R} {\cal P}(\lambda) \in {\cal M}^{s}_{\sigma_1,\sigma_3},$ and
	\begin{align*}
	\normaSsigsig{\Gradbeta {\cal R} {\cal P}}{\sigma_1}{\sigma_3} \lesssim_{s, \beta}
	\| \Gradbeta {\cal R} \|^{\Lip}\indSsigsig{s}{\sigma_2}{ \sigma_3} \| {\cal P} \|^{\Lip}
	\indSsigsig{s}{\sigma_1}{ \sigma_2}
	+ \| {\cal R} \|^{\Lip}\indSsigsig{s}{\sigma_2}{ \sigma_3} \|\Gradbeta {\cal P} \|^{\Lip}\indSsigsig{s}{\sigma_1}{ \sigma_2}.
	\end{align*}
\end{lemma}

\begin{proof}
	{\sc Estimate $(i)$.}
	By applying Lemma \ref{lemma prodotto 1}-$(i)$, one computes
	\begin{align*}
	\left(\|({\cal  R} {\cal P}) \|\indSsigsig{s}{\s_1}{\s_3}\right)^2 &\leq
	\sum_{l \in \Z^{\frak n}}  \langle l \rangle ^{2s}
	\left[\sum_{l' \in \Z^{\frak n}} \normahs{\widehat{{\cal  R}}(l-l')}{\s_2}{ \s_2} \normahs{\widehat{{\cal P}}(l')}{\s_1}{ \s_2}  \right]^2\\
	&\lesssim_{s} \sum_{l \in \Z^{\frak n}} 
	\left[\sum_{l' \in \Z^{\frak n}}  \langle l' \rangle ^{s} \normahs{\widehat{{\cal R}}(l-l')}{\s_2}{ \s_3} \normahs{\widehat{{\cal P}}(l')}{\s_1}{ \s_2}  \right]^2\\
	&+  \sum_{l \in \Z^{\frak n}} 
	\left[\sum_{l' \in \Z^{\frak n}}  \langle l- l' \rangle ^{s} \normahs{\widehat{{ \cal R}}(l-l')}{\s_2}{ \s_3} \normahs{\widehat{{\cal P}}(l')}{\s_1}{ \s_2}  \right]^2\\
	&\leq \sum_{l, l' \in Z^{\frak n}} \langle l' \rangle ^{2s} \langle l- l' \rangle ^{2s_0} \left(\normahs{\widehat{{\cal R}}(l-l')}{\s_2}{ \s_3}\right)^{2}
	\left(\normahs{\widehat{{\cal P}}(l')}{\s_1}{ \s_2}\right)^{2}\\
	&+ \sum_{l, l' \in Z^{\frak n}} \langle l' \rangle ^{2s_0} \langle l- l' \rangle ^{2s} \left(\normahs{\widehat{{\cal R}}(l-l')}{\s_2}{ \s_3}\right)^{2}
	\left(\normahs{\widehat{{\cal P}}(l')}{\s_1}{ \s_2}\right)^{2}\\
	&=\left(\| {\cal R} \| \indSsigsig{s_0}{\s_2}{ \s_3}\right)^{2} \left(\| {\cal P} \|\indSsigsig{s}{\s_1}{ \s_2}\right)^{2} +\left(\| {\cal R} \|\indSsigsig{s}{\s_2}{\s_3}\right)^{2} \left(\| {\cal P} \|\indSsigsig{s_0}{\s_1}{ \s_2}\right)^{2}.
	\end{align*}
	To get the required estimate in Lipschitz norm, it is sufficient to decompose
	$$
	\left({\cal R} {\cal P}\right)(\lambda_2) - \left({\cal R} {\cal P}\right)(\lambda_1) = {\cal R}(\lambda_2) \left( {\cal P}(\lambda_2) - {\cal P}(\lambda_1)\right) + \left({\cal R}(\lambda_2) - {\cal R}(\lambda_1)\right){\cal P}(\lambda_1)
	$$
	and to apply the above inequality to both the terms of the right-hand side, taking respectively
	$$ {\cal R}(\lambda_2) \textrm{ as } {\cal R}, \quad {\cal P}(\lambda_2) - {\cal P}(\lambda_1) \textrm{ as } {\cal P}$$
	and
	$$ {\cal R}(\lambda_2) - {\cal R}(\lambda_1) \textrm{ as } {\cal R}, \quad\ {\cal P}(\lambda_1)  \textrm{ as } {\cal P}.$$
	{\sc Estimate $(ii)$.}
	Arguing as before, one has
	\begin{align}
	\left(\|\Gradbeta ({\cal R} {\cal P}) \|\indSsigsig{s}{\s_1}{\s_3}\right)^2 &= \sum_{l \in Z^p} \langle l \rangle ^{2s} \big(\normahs{ \Gradbeta \widehat{({\cal R} {\cal P})}(l)} {\s_1}{\s_3}\big)^2 \nonumber\\
	&\leq \sum_{l \in \Z^{\frak n}} \langle l \rangle ^{2s} \left( \sum_{l' \in \Zp} \normahs{ \Gradbeta \widehat{{\cal R}}(l-l') \widehat{{\cal P}}(l')}{\s_1}{ \s_3}\right)^2 \nonumber\\
	& \lesssim_{s} \sum_{l, l' \in Zp} \langle l' \rangle ^{2s} \langle l- l' \rangle ^{2s}  \left(\normahs{ \Gradbeta \widehat{{\cal R}}(l-l') \widehat{{\cal P}}(l') }{\s_1}{ \s_3}\right)^2 \label{abatantuono 1}
	\end{align}
	where in the last inequality, we have used that 
	$$
	\langle l \rangle^{2 s} \lesssim_s \langle l'\rangle^{2 s} + \langle l - l' \rangle^{2 s} \lesssim_s \langle l'\rangle^{2 s} \langle l - l' \rangle^{2 s}\,. 
	$$
	By applying Lemma \ref{lemma prodotto 1}-$(ii)$ (to estimate $\normahs{ \Gradbeta \widehat{{\cal R}}(l-l') \widehat{{\cal P}}(l') }{\s_1}{ \s_3}$) one obtains that 
		\begin{align*}
	\left(\|\Gradbeta ({\cal R} {\cal P}) \|\indSsigsig{s}{\s_1}{\s_3}\right)^2 &\lesssim_{s, \beta} \sum_{l, l' \in Zp} \langle l' \rangle ^{2s} \langle l- l' \rangle ^{2s} \left(\normahs{\Gradbeta \widehat{{\cal R}}(l-l')}{\s_2}{ \s_3}\right)^2 \left(\normahs{\widehat{{\cal P}}(l')}{\s_1}{ \s_2}\right)^2\\
	&+ \sum_{l, l' \in Zp} \langle l' \rangle ^{2s} \langle l- l' \rangle ^{2s} \left(\normahs{\widehat{{\cal R}}(l-l')}{\s_2}{ \s_3}\right)^2 \left(\normahs{\Gradbeta \widehat{{\cal P}}(l')}{\s_1}{ \s_2}\right)^2\\
	& \lesssim_{s, \beta} \left(\|\Gradbeta {\cal R}\|\indSsigsig{s}{\s_2}{ \s_3}\right)^{2} \left(\|{\cal P} \| \indSsigsig{s}{\s_1}{ \s_2}\right)^{2}  \\
	& \quad + \left(\|{\cal R}\|\indSsigsig{s}{\s_2}{ \s_3}\right)^{2} \left(\| \Gradbeta {\cal P} \| \indSsigsig{s}{\s_1}{ \s_2}\right)^{2}.
	\end{align*}
	Concerning the Lipschitz estimates, as in the proof of $(i)$ we write
	\begin{align*}
	\Gradbeta \left({\cal R}{\cal P}(\lambda_2) - {\cal R}{\cal P}(\lambda_1)\right) &= \Gradbeta 
	{\cal R}(\lambda_2) \left( {\cal P}(\lambda_2) - {\cal P}(\lambda_1)\right) + \Gradbeta \left({\cal R}(\lambda_2) - {\cal R}(\lambda_1)\right){\cal P}(\lambda_1)
	\end{align*}
	and we repeat the same argument with
	$$ {\cal R}(\lambda_2) \textrm{ as } {\cal R}, \quad {\cal P}(\lambda_2) - {\cal P}(\lambda_1) \textrm{ as } {\cal P}$$
	and
	$$ {\cal R}(\lambda_2) - {\cal R}(\lambda_1) \textrm{ as } {\cal R}, \quad\ {\cal P}(\lambda_1)  \textrm{ as } {\cal P}.$$
\end{proof}
Iterating the estimates of Lemma \ref{stima composizione}, one gets for any $s \geq s_0$, $\sigma \in \R$, $n \geq 1$
\begin{equation}\label{stime composizione iterate}
\begin{aligned}
& \| {\cal R}^n  \|^{\rm Lip}_{{\cal M}^s_{\sigma , \sigma }} \leq C(s)^n \| {\cal R} \|^{\rm Lip}_{{\cal M}^s_{\sigma , \sigma }} \big( \| {\cal R} \|^{\rm Lip}_{{\cal M}^{s_0}_{\sigma , \sigma}} \big)^{n - 1}\,, \\
& \| \langle \nabla \rangle^\beta({\cal R}^n) \|_{{\cal M}^s_{\sigma , \sigma }}^{\rm Lip} \leq C(s, \beta)^n  \| \langle \nabla \rangle^\beta{\cal R} \|_{{\cal M}^s_{\sigma , \sigma }}^{\rm Lip}  \big(\| {\cal R} \|_{{\cal M}^s_{\sigma , \sigma}}^{\rm Lip} \big)^{n - 1}\,. 
\end{aligned}
\end{equation}
The following lemma holds:
\begin{lemma}\label{Neumann series}
Let $s \geq s_0$, $\sigma \in \R$, $\beta \geq 0$ and $X(\lambda), \langle \nabla \rangle^\beta X(\lambda) \in {\cal M}^s_{\sigma, \sigma}$. Then there exists $\delta (s, \beta) \in (0, 1)$ such that if $\| X\|_{{\cal M}^s_{\sigma, \sigma}}^{\rm Lip} \leq  \delta(s, \beta)$, then $\Phi := {\rm Id} + X$ is invertible and its inverse $\Phi^{- 1}$ satisfies the estimates 
$$
\| \Phi^{- 1} - {\rm Id} \|_{{\cal M}^s_{\sigma, \sigma}}^{\rm Lip} \lesssim_s \| X\|_{{\cal M}^s_{\sigma, \sigma}}^{\rm Lip}\,, \quad \| \langle \nabla \rangle^\beta (\Phi^{- 1} - {\rm Id}) \|_{{\cal M}^s_{\sigma, \sigma}}^{\rm Lip} \lesssim_s \| \langle \nabla \rangle^\beta X\|_{{\cal M}^s_{\sigma, \sigma}}^{\rm Lip}\,. 
$$
\end{lemma}
\begin{proof}
By the Neumann series one has $\Phi^{- 1} - {\rm Id} = \sum_{n \geq 1}(- 1)^n X^n$. Then, applying the estimates \eqref{stime composizione iterate} to each term $X^n$, the claimed statement follows. 
\end{proof}

\subsection{Other estimates in ${\cal M}^{s}_{\s_1,\s_2}$}

\begin{lemma} \label{lemma norma hs norma op}
$(i)$	Let $\sigma_1, \sigma_2 \in \R$ and $A \in {\cal B} \left( {\cal H}^{\s_1 - \eta}, {\cal H}^{\s_2}\right),$ $\eta > \frac{d}{2}$, then
	$$
	\normahs{A}{\sigma_1}{ \sigma_2} \lesssim_{\eta} \| A \|_{{\cal B}({\cal H}^{\sigma_1 - \eta}, {\cal H}^{\sigma_2} )},
	$$
	
	\noindent
	$(ii)$ Let $\sigma_1, \sigma_2 \in \R$, $\beta \geq 0$, $\eta > \frac{d}{2}$. Then if $A \in {\cal B}({\cal H}^{\sigma_1 - \beta - \eta}, {\cal H}^{\sigma_2 + \beta})$, one has  
	$$
	\| \langle \nabla \rangle^\beta A \|^{HS}_{\sigma_1, \sigma_2} \lesssim_\beta  \| A \|_{{\cal B}({\cal H}^{\sigma_1 - \beta - \eta}, {\cal H}^{\sigma_2 + \beta})}\,. 
	$$
\end{lemma}
\begin{proof}
{\sc Proof of $(i)$.}
	 Let us consider $\forall\ k' \in \Zd$ $u^{(k')} \in {\cal H}^{\s_1}$ defined by
	$$
	\hat{u}^{(k')}_{h} = \begin{cases}
	
	\langle k' \rangle ^{-(\s_1 - \eta)} \quad \textrm{if }h=k'\\
	0 \quad \quad \textrm{if }h\neq k';
	\end{cases}
	$$
	We have that
	\begin{align*}
	\sum_{k} \langle k \rangle^{2 \s_2} |A_{k}^{k'}|^2 \langle k' \rangle ^{-2(\s_1 - \eta)}&= \|A u^{(k')} \|^2 _{{\cal H}^{\s_2}}\\
	&\leq \| A\|^2 _{\bSunoSdue{\sigma_1-\eta}{\sigma_2}} \| u^{(k')}\|^2_{{\cal H}^{\s_1-\eta}}\\
	&=\| A\|^2_{\bSunoSdue{\sigma_1-\eta}{\sigma_2}},
	\end{align*}
	since $\| u^{(k')}\|_{{\cal H}^{\s_1-\eta}} = 1.$ Thus we deduce that $\forall\ k'$
	\begin{equation} \label{stima improbabile}
	\sum_{k} \langle k \rangle^{2 \s_2} |A_{k}^{k'}|^2 \leq \| A\|^2_{\bSunoSdue{\sigma_1-\eta}{\sigma_2}} \langle k' \rangle ^{2(\s_1 - \eta)}.
	\end{equation}
	Let now $u$ be a generic function in ${\cal H}^{\s_1}:$ from \eqref{stima improbabile} it follows that
	\begin{align*}
	\left(\normahs{A}{\sigma_1}{\sigma_2}\right)^2 &= \sum_{k, k' \in \Z^d} \langle k \rangle^{2 \sigma_2} |A_{k}^{k'}|^2 \langle k'\rangle^{-2\sigma_1}\\
	&\leq \sum_{k' \in \Z^d} \langle k' \rangle ^{2(\sigma_1 -\eta)} \| A\|^2_{\bSunoSdue{\sigma_1-\eta}{\sigma_2}} \langle k' \rangle ^{-2 \sigma_1}\\
	&\lesssim_{\s_0} \| A\|^2_{\bSunoSdue{\sigma_1-\eta}{\sigma_2}}.
	\end{align*}
	
	\noindent
	{\sc Proof of $(ii)$.} Using that for any $j, j' \in \Z^d$, $\langle j - j' \rangle^{2 \beta} \lesssim_\beta \langle j \rangle^{2 \beta} + \rangle j' \rangle^{2 \beta} \lesssim_\beta \langle j \rangle^{2 \beta} \langle j' \rangle^{2 \beta}$, one gets that 
	\begin{align}
	\big(\| \langle \nabla \rangle^\beta A\|^{HS}_{\sigma_1, \sigma_2}\big)^2 & = \sum_{j, j' \in \Z^d} \langle j \rangle^{2 \sigma_2} \langle j - j' \rangle^{2 \beta} |A_j^{j'}|^2 \langle j' \rangle^{- 2 \sigma_1} \nonumber\\
	& \lesssim_\beta  \sum_{j, j' \in \Z^d} \langle j \rangle^{2 (\sigma_2 + \beta)}  |A_j^{j'}|^2 \langle j' \rangle^{- 2 (\sigma_1 - \beta)} = \big( \| A\|^{HS}_{\sigma_1 - \beta, \sigma_2 + \beta} \big)^2\,. 
	\end{align}
	The claimed statement follows by applying item $(i)$ (replacing $\sigma_1$ with $\sigma_1 - \beta$ and $\sigma_2$ with $\sigma_2 + \beta$). 
\end{proof}

\begin{lemma} \label{lemma legame tra norme}
$(i)$	Let $A \in {\cal C}^s\Big(\T^{\frak n}; {\cal B} \left( {\cal H}^{\s_1 - \eta},\ {\cal H}^{\s_2}\right) \Big),\ \s_1, \sigma_2 \in \R$, $\eta > \frac{d}{2}$ and $s \geq 0$. Then 
	$$
	\| A \|\indSsigsig{s}{\sigma_1}{\sigma_2}  \lesssim \| A \|_{{\cal C}^s \big( \T^{\frak n}; {\cal B} \left( {\cal H}^{\s_1 - \eta},\ {\cal H}^{\s_2}\right) \big) }.
	$$
	
	\noindent
	$(ii)$ Let $s \geq 0$, $\sigma_1, \sigma_2 \in \R$, $\beta \geq 0$, $\eta > \frac{d}{2}$ and $A \in {\cal C}^s \Big( \T^{\frak n};  {\cal B}({\cal H}^{\sigma_1 - \beta - \eta}, {\cal H}^{\sigma_2 + \beta}) \Big)$. Then 
	$$
	\| \langle \nabla \rangle^\beta A\|_{{\cal M}^s_{\sigma_1, \sigma_2}} \lesssim_{ \beta}  \| A\|_{{\cal C}^s \Big( \T^{\frak n};  {\cal B}({\cal H}^{\sigma_1 - \beta - \eta}, {\cal H}^{\sigma_2 + \beta})\Big)}
	$$
	\end{lemma}
\begin{proof}
The claimed statement follows recalling that ${\cal M}_{\sigma_1, \sigma_2}^s = {\cal H}^s\Big( \T^{\frak n} ;  {\cal B}^{HS}({\cal H}^{\sigma_1}, {\cal H}^{\sigma_2}) \Big)$, by applying Lemma \ref{lemma norma hs norma op} and using that for every Banach space $X$ one has that $\| \cdot \|_{H^s(\T^{\frak n};  X)} \leq \| \cdot \|_{{\cal C}^s(\T^{\frak n};  X)}$. 
\end{proof}

\begin{lemma} \label{lemma norma dm}
$(i)$	Let $m \geq 0$, $A \in {\cal C}^\infty(\T^{\frak n}, OPS^{-\kappa}),\ \kappa > 2 m + \frac{d}{2}$. Then for any $\sigma \in \R$, 

$A \in {\cal C}^\infty(\T^{\frak n};  {\cal B} \left( {\cal H}^{\s +m-\kappa}, {\cal H}^{\s+m}\right))$ and for any $s \geq 0$
	$$
	\|A \|\indSsigsig{s}{\s-m}{\s +m} \lesssim \| A \|_{{\cal C}^s(\T^{\frak n} ;  {\cal B}({\cal H}^{\s + m - \kappa}, {\cal H}^{\s + m}))}\, 
		$$
		
		\noindent
		$(ii)$ Let $m, \beta \geq 0$ and $A \in {\cal C}^\infty(\T^{\frak n} ;  OPS^{- \kappa})$, $\kappa > 2 m + 2 \beta + \frac{d}{2}$. Then for any $\sigma \in \R$, 
		
		$A \in {\cal C}^\infty \Big(\T^{\frak n};  {\cal B}\big({\cal H}^{\sigma + m + \beta - \kappa}, {\cal H}^{\sigma + m + \beta} \big) \Big)$ and for any $s \geq 0$
		$$
		\| \langle \nabla \rangle^\beta A\|_{{\cal M}^s_{\sigma - m, \sigma + m}} \lesssim_\beta \| A \|_{{\cal C}^s(\T^{\frak n};  {\cal B}({\cal H}^{\s + m + \beta - \kappa}, {\cal H}^{\s + m + \beta}))}\,. 
		$$
\end{lemma}
\begin{proof}
	The statement $(i)$ follows by applying Lemma \ref{lemma legame tra norme}-$(i)$ with $\sigma_1= \s -m,\ \sigma_2=\s +m,\ \eta=\kappa-2m$. 
	
	\noindent
	The statement $(ii)$ follows by applying Lemma \ref{lemma legame tra norme}-$(ii)$ with $\sigma_1= \s -m,\ \sigma_2=\s +m,\ \eta=\kappa-2m - 2 \beta$.
	\end{proof}
\begin{lemma} \label{lemma decay operatore smooth}
	Let $\s \in \R$, $\kappa \geq 0$, $ P(\lambda) \in {\cal B}^{HS}({\cal H}^\sigma, {\cal H}^{\sigma + \kappa})$, $\lambda \in \Omega_o \subseteq \R^{\frak n + d}$. Then $\forall  j \in \Z^{d}$ its matrix elements $P_j^{j}$ satisfy
	$$
	|P_j^{j}|\leq \|P\|^{HS}_{\s, \s + \kappa}  \langle j \rangle ^{-\kappa}, \quad |P_j^{j}|^{\rm Lip}\leq \|P\|^{HS, \rm Lip}_{\s, \s + \kappa}  \langle j \rangle ^{-\kappa}.
	$$
\end{lemma}
\begin{proof}
For any $j \in \Z^d$, one has 
\begin{align*}
(\| P \|^{HS}_{\sigma, \sigma + \kappa})^2 & = \sum_{k, k' \in \Z^d} \langle k \rangle^{2 (\sigma + \kappa)} |P_k^{k'}| \langle k' \rangle^{- 2 \sigma} \geq \langle j \rangle^{2 (\sigma + \kappa)} |P_j^j| \langle j \rangle^{- 2 \sigma} = \langle j \rangle^{\kappa} |P_j^j|\,.
\end{align*}
The Lipschitz estimate follows arguing similarly by estimating $\frac{\| P(\lambda_1) - P(\lambda_2 \|_{\sigma, \sigma + \kappa}^{HS}}{|\lambda_1 - \lambda_2|}$ for any $\lambda_1, \lambda_2 \in \Omega_o$, $\lambda_1 \neq \lambda_2$. 
\end{proof}

\def\cprime{$'$}


\end{document}


\begin{thebibliography}{DvStvc96}

\bibitem[ABHK18]{controlloAB}
Thomas Alazard, Pietro Baldi, and Daniel Han-Kwan.
\newblock Control of water waves.
\newblock {\em J. Eur. Math. Soc. (JEMS)}, 20(3):657--745, 2018.

\bibitem[Bam17]{Bam17}
D.~Bambusi.
\newblock Reducibility of 1-d {S}chr\"odinger equation with time quasiperiodic
  unbounded perturbations, {II}.
\newblock {\em Comm. Math. Phys.}, 353(1):353--378, 2017.

\bibitem[Bam18]{Bam18}
D.~Bambusi.
\newblock Reducibility of 1-d {S}chr\"odinger equation with time quasiperiodic
  unbounded perturbations, {I}.
\newblock {\em Trans. Amer. Math. Soc.}, 370(3):1823--1865, 2018.

\bibitem[BBHM17]{BBHM}
Pietro Baldi, Massimiliano Berti, Emanuele Haus, and Riccardo Montalto.
\newblock Time quasi-periodic gravity water waves in finte depth.
\newblock {\em Inventiones Math.}, To appear. Preprint arXiv:1708.01517,
 2017.

\bibitem[BBM14]{BBM14}
P.~Baldi, M.~Berti, and R.~Montalto.
\newblock K{AM} for quasi-linear and fully nonlinear forced perturbations of
  {A}iry equation.
\newblock {\em Math. Ann.}, 359(1-2):471--536, 2014.

\bibitem[BBM16a]{BBM-auto}
Pietro Baldi, Massimiliano Berti, and Riccardo Montalto.
\newblock K{AM} for autonomous quasi-linear perturbations of {K}d{V}.
\newblock {\em Ann. Inst. H. Poincar\'e Anal. Non Lin\'eaire},
  33(6):1589--1638, 2016.

\bibitem[BBM16b]{BBM-mKdV}
Pietro Baldi, Massimiliano Berti, and Riccardo Montalto.
\newblock K{AM} for autonomous quasi-linear perturbations of m{K}d{V}.
\newblock {\em Boll. Unione Mat. Ital.}, 9(2):143--188, 2016.

\bibitem[Bel85]{Bel}
Jean Bellissard.
\newblock Stability and instability in quantum mechanics.
\newblock In {\em Trends and developments in the eighties ({B}ielefeld,
  1982/1983)}, pages 1--106. World Sci. Publishing, Singapore, 1985.

\bibitem[BFH17]{controlloBHF}
Pietro Baldi, Giuseppe Floridia, and Emanuele Haus.
\newblock Exact controllability for quasilinear perturbations of {K}d{V}.
\newblock {\em Anal. PDE}, 10(2):281--322, 2017.

\bibitem[BG01]{BG01}
D.~Bambusi and S.~Graffi.
\newblock Time quasi-periodic unbounded perturbations of {S}chr\"odinger
  operators and {KAM} methods.
\newblock {\em Comm. Math. Phys.}, 219(2):465--480, 2001.

\bibitem[BGMR17]{BGMR2}
D.~{Bambusi}, B.~{Grebert}, A.~{Maspero}, and D.~{Robert}.
\newblock {Growth of Sobolev norms for abstract linear {S}chr\"odinger
  Equations}.
\newblock {\em arXiv:1706.09708}, 2017.

\bibitem[BGMR18]{BGMR1}
D.~{Bambusi}, B.~{Grebert}, A.~{Maspero}, and D.~{Robert}.
\newblock {Reducibility of the quantum Harmonic oscillator in $d$-dimensions
  with polynomial time dependent perturbation}.
\newblock {\em Analysis \& {PDE}s}, 11(3):775--799, 2018.

\bibitem[BHM18]{controlloBHM}
Pietro Baldi, Emanuele Haus, and Riccardo Montalto.
\newblock Controllability of quasi-linear {H}amiltonian {NLS} equations.
\newblock {\em J. Differential Equations}, 264(3):1786--1840, 2018.

\bibitem[BM16]{BertiMontalto} M. Berti, R. Montalto, \emph{Quasi-periodic standing wave solutions for gravity-capillary water waves}, to appear in Memoirs of the Amer. Math. Society, MEMO 891. Preprint arXiv:1602.02411v1, 2016. 


\bibitem[Com87]{C87}
M.~Combescure.
\newblock The quantum stability problem for time-periodic perturbations of the
  harmonic oscillator.
\newblock {\em Ann. Inst. H. Poincar\'e Phys. Th\'eor.}, 47(1):63--83, 1987.

\def\cek{cek}
\bibitem[DS96]{DS96}
P.~Duclos and P.~S\v tov\'\i~\v\cek.
\newblock Floquet {H}amiltonians with pure point spectrum.
\newblock {\em Comm. Math. Phys.}, 177(2):327--347, 1996.

\bibitem[EK09]{EK09}
H.~L. Eliasson and S.~B. Kuksin.
\newblock On reducibility of {S}chr\"odinger equations with quasiperiodic in
  time potentials.
\newblock {\em Comm. Math. Phys.}, 286(1):125--135, 2009.

\bibitem[FGMP18]{trasporto paper} R. Feola, F. Giuliani, R. Montalto, M. Procesi, \emph{Reducibility of first order linear operators on tori via Moser's theorem.} Preprint  arXiv:1801.04224, 2018.


\bibitem[FP15]{FP15}
R.~Feola and M.~Procesi.
\newblock Quasi-periodic solutions for fully nonlinear forced reversible
  {S}chr\"odinger equations.
\newblock {\em J. Differential Equations}, 259(7):3389--3447, 2015.

\bibitem[GP16]{GP}
B.~{Gr{\'e}bert} and E.~{Paturel}.
\newblock {On reducibility of quantum Harmonic Oscillator on $\mathbb{R}^d$
  with quasiperiodic in time potential}.
\newblock {\em  Ann. Fac. Sci. Toulouse Math.}, to appear. Preprint arXiv:1603.07455.

\bibitem[IPT05]{IPT05}
G.~Iooss, P.~I. Plotnikov, and J.~F. Toland.
\newblock Standing waves on an infinitely deep perfect fluid under gravity.
\newblock {\em Arch. Ration. Mech. Anal.}, 177(3):367--478, 2005.

\bibitem[KP03]{KapPoe}
T.~Kappeler and J.~P{\"o}schel.
\newblock {\em {KAM} \& {K}d{V}}.
\newblock Springer, 2003.

\bibitem[Kuk87]{K87}
S.~B. Kuksin.
\newblock Hamiltonian perturbations of infinite-dimensional linear systems with
  an imaginary spectrum.
\newblock {\em Funct. Anal. Appl.}, 21:192--205, 1987.

\bibitem[Kuk93]{Kuk93}
S.~B. Kuksin.
\newblock {\em Nearly integrable infinite-dimensional {H}amiltonian systems},
  volume 1556 of {\em Lecture Notes in Mathematics}.
  \newblock Springer-Verlag, Berlin, 1993.

  
\bibitem[Kuk97]{Kuk97}
S.~B. Kuksin.
\newblock On small-denominators equations with large variable coefficients.
\newblock {\em Z. Angew. Math. Phys.}, 48(2):262--271, 1997.

\bibitem[LY10]{LY10}
J.~Liu and X.~Yuan.
\newblock Spectrum for quantum {D}uffing oscillator and small-divisor equation
  with large-variable coefficient.
\newblock {\em Comm. Pure Appl. Math.}, 63(9):1145--1172, 2010.


\bibitem[Mon17a]{Mon17a} R. Montalto, \emph{A reducibility result for a class of linear wave equations on $\T^d$.} Int. Math. Res. Notices, doi: 10.1093/imrn/rnx167, 2017.  


\bibitem[Mon18]{Mon18}
Riccardo Montalto.
\newblock On the growth of {S}obolev norms for a class of linear
  {S}chr\"odinger equations on the torus with superlinear dispersion.
\newblock {\em Asymptot. Anal.}, 108(1-2):85--114, 2018.

\bibitem[Mon18a]{Montalto2} R. Montalto, \emph{Growth of Sobolev norms for time dependent Schr\"odinger equations with sublinear dispersion.} Preprint arXiv:1802.04138v1, 2018. 
	


\bibitem[MR17]{MaRo}
A.~Maspero and D.~Robert.
\newblock On time dependent {S}chrödinger equations: {G}lobal well-posedness
  and growth of {S}obolev norms.
\newblock {\em Journal of Functional Analysis}, 273(2):721 -- 781, 2017.

\bibitem[Tay91]{Taylor}
Michael~E. Taylor.
\newblock {\em Pseudodifferential operators and nonlinear {PDE}}, volume 100 of
  {\em Progress in Mathematics}.
\newblock Birkh\"auser Boston, Inc., Boston, MA, 1991.

\end{thebibliography}

\begin{thebibliography}{10}
	
	
	\bibitem{BM16} 
	M. Berti, R. Montalto, 
	\newblock \emph{Quasi-periodic water waves}. J. Fixed Point Theory Appl., 
	19, no. 1, 129-156, 2017.
	
	
	

		
		
 

	
	
	
	
	
	
	
\end{thebibliography}
\end{document}